\documentclass[journal]{IEEEtran}
\IEEEoverridecommandlockouts

\usepackage{cite}
\usepackage{amsmath,amssymb,amsfonts, amsthm}
\usepackage{mathtools}
\usepackage{algorithmic}
\usepackage{graphicx}
\usepackage{textcomp}
\usepackage{xcolor}
\def\BibTeX{{\rm B\kern-.05em{\sc i\kern-.025em b}\kern-.08em
    T\kern-.1667em\lower.7ex\hbox{E}\kern-.125emX}}

\colorlet{revisioncolor}{red}

\def\BibTeX{{\rm B\kern-.05em{\sc i\kern-.025em b}\kern-.08em
    T\kern-.1667em\lower.7ex\hbox{E}\kern-.125emX}}
    
\usepackage{mathrsfs}
\usepackage{url}

\usepackage{tikz}

 \usepackage[caption=false,font=footnotesize]{subfig}

\newtheorem{Prop}{Proposition}[]
\newtheorem{Teo}{Theorem}[]
\newtheorem{Cor}{Corollary}

\newtheorem{Lema}{Lemma}[]

\newcommand{\Z}{\mathbb{Z}}

\newcommand{\R}{\mathbb{R}}

\newcommand{\h}{\mathbb{H}}

\newcommand{\s}{\mathscr{S}}

\begin{document}

\title{Shaping and Finite-Blocklength Analysis of Quaternionic Alamouti-Type Codes}

\author{Juliana~G. F. Souza,~\IEEEmembership{Member, ~IEEE,}
        and~Yu-Chih~Huang,~\IEEEmembership{Senior Member,~IEEE}
\thanks{Manuscript received Month XX, 2026.  This work was supported in part by the Sao Paulo Research Foundation (FAPESP) under Grant 2025/22940-2, by the National Science and Technology Council, Taiwan, under Grant NSTC 113-2223-E-A49-005-MY3 and a grant of access to Open AI models through the ChatGPT for Academic Researchers program. An earlier version of parts of this work was submitted at the 2026 IEEE International Symposium on Information Theory.}%
\thanks{Juliana G. F. Souza is with the Institute of Mathematics, Statistics and Scientific Computing (IMECC), University of Campinas (Unicamp), Campinas, São Paulo, Brazil (e-mail: julianagfs@ime.unicamp.br).}
\thanks{Yu-Chih Huang is with the Institute of Communications Engineering,
National Yang Ming Chiao Tung University, Hsinchu, Taiwan (e-mail: jerryhuang@nycu.edu.tw).}}

\maketitle

\begin{abstract}
We study finite $2\times2$ Alamouti-type space--time codes associated with classical maximal orders in definite quaternion algebras over $\mathbb Q$. For the family indexed by primes $q\equiv3\pmod4$, we prove that the Eisenstein case $q=3$ uniquely minimises the finite quotient-alphabet energy at fixed minimum lattice distance and alphabet size. This ordering holds for every integer $p\geq2$, including when the representatives and their common translation are jointly optimised, with an explicit positive gap for every $q>3$. For uniform finite inputs, we derive a controlled high-SNR expansion of the information variance and determine its limiting ratio to the mutual information deficit. The mutual information expansion specialises arbitrary input fading results and gives an explicit remainder. Standard random-coding and second-order achievability tools relate these quantities to reliability over independently faded Alamouti blocks. All alphabets preserve symbol-wise maximum-likelihood detection at the Alamouti level. Numerical energy, mutual information, and error-rate comparisons support the analysis. Centred comparisons with classical Alamouti, Golden, and Silver codes illustrate the trade-off between reliability, symbol rate, and decoding effort.
\end{abstract}


\section{Introduction}
\label{sec:introduction}

Space--time block codes (STBCs) provide transmit diversity in multiple-input multiple-output (MIMO) systems without requiring a large increase in receiver complexity. The $2\times2$ Alamouti code is the main orthogonal example. It achieves full diversity and permits exact symbol-wise maximum likelihood (ML) detection because its two columns are orthogonal~\cite{alamouti1998,tarokh1999stbc,paulraj2003}. In its usual form, the information symbols are selected from square quadrature-amplitude modulation (QAM) constellations obtained from the Gaussian integers $\mathbb Z[\mathrm i]$.

The Alamouti code also admits an algebraic interpretation. Its code matrices can be obtained from a quaternion algebra, and the determinant of a matrix is the reduced norm of the corresponding quaternion element~\cite{oggier2007cyclic}. More generally, division algebras, natural orders, and maximal orders have been used to construct fully diverse space--time lattices with non-vanishing determinants and controlled decoding properties~\cite{oggier2006perfect,elia2006explicit,hollanti2008maximal,vehkalahti2009densest,alves2015lattices}. These works show that algebraic orders can control the determinant structure, while their Euclidean embeddings determine the geometry of the transmitted constellation.

The choice of finite alphabet is closely related to constellation shaping. Geometric shaping changes the positions of constellation points, while probabilistic shaping changes their transmission probabilities. Constellation shaping has been discussed as a possible direction for future mobile systems~\cite{itur_m2516}, and both approaches have been considered in 3GPP contributions in ongoing 6G standardisation activities~\cite{Samsung2026Modulation6GR}.

Integrated sensing and communication (ISAC) is recognised as one of the main usage scenarios for 6G mobile systems~\cite{itur_m2160}. Constellation shaping is particularly attractive for ISAC, because the constellation affects both information transmission and sensing performance. In~\cite{Shao2026APSKISAC}, geometric shaping of amplitude and phase shift keying (APSK) constellations is studied, using their ring structure and energy allocation to control the communication sensing trade-off. Probabilistic shaping is used in~\cite{Du2024ISACPCS} to balance achievable information rate and sensing performance in orthogonal frequency division multiplexing (OFDM) systems. Joint optimisation of constellation geometry and symbol probabilities is proposed in~\cite{Geiger2025Joint} for OFDM-ISAC, considering both communication performance and target detection. Motivated by these developments, this work explores geometric constellation shaping for space--time codes, using uniform finite alphabets in a quaternionic Alamouti family.

The $2\times2$ Alamouti matrix accepts any complex input alphabet. Hence, replacing square QAM by a hexagonal alphabet does not by itself define a new orthogonal matrix design. Eisenstein constellations and embedded Alamouti structures have previously appeared in modulation and space--time coding~\cite{freudenberger2017eisenstein,oggier2007cyclic,sinnokrot2008alamouti}.

Eisenstein based full rate algebraic STBCs are also available~\cite{wang2012quadratic}. The practical advantage considered here is the retention of exact symbol-wise ML decoding for every coefficient alphabet. Each Alamouti block carries two complex symbols over two channel uses, while the Golden and Silver codes carry four~\cite{belfiore2005golden,hollanti2008silver}. Their fast ML decoders reduce the search complexity but retain coupled symbol decisions~\cite{sinnokrot2010golden,srinath2011silver}. Thus the orthogonal family offers simpler symbol separation at the cost of a lower symbol rate. Orthogonality and this decoding structure are inherited from the Alamouti design. Our contribution concerns the finite alphabet shaping and information theoretic analysis in the quaternionic family.

In our earlier work~\cite{souza2026AESTBC}, we considered the Eisenstein case obtained from $(-1,-3)_{\mathbb Q}$ and studied shaping, mutual information, and finite-blocklength performance. That work already included the variance contribution caused by shared fading and a normal approximation. The present paper establishes exact finite energy comparisons throughout the prime indexed family, including joint optimisation of representatives and translation. It also replaces the earlier distance based mutual information approximation by a controlled fixed-alphabet expansion and gives a controlled information variance expansion with its limiting ratio to the mutual information deficit.

Finite-blocklength information theory describes reliability when the packet length is limited~\cite{polyanskiy2010finite}. Bounds for multiple-antenna Rayleigh block-fading channels were studied in~\cite{durisi2016short-packet}. In~\cite{collins2019coherent}, the coherent channel dispersion was derived, and the role of orthogonal designs, including Alamouti, was explained for the MISO setting with Gaussian signalling. Arbitrary input mutual information expansions over fading channels were studied in~\cite{rodrigues2014multiple,ramos2014fading}. We specialise this mutual information framework to uniform finite alphabets with an explicit remainder. For the information variance, we control the high-SNR remainder and determine the limiting ratio to the mutual information deficit, retaining the contribution of the fading shared by the two Alamouti symbols. Standard achievability tools then connect these fixed-input quantities to finite-blocklength reliability.

We consider the imaginary quadratic fields $K_q=\mathbb Q(\sqrt{-q})$, where $q\equiv3\pmod4$ is prime, and their rings of integers $\mathcal O_q$. For each $q$, the definite quaternion algebra $B_q=K_q\oplus uK_q, u^2=-1, ux=\bar xu,$ contains the classical maximal order $\Gamma_q=\mathcal O_q\oplus u\mathcal O_q$~\cite{voight2021quaternion}. Its complex embedding gives an Alamouti-type code whose coefficient lattice is denoted by $\Lambda_q$. All members have the same orthogonal transmission structure, diversity order, and symbol-wise ML detector. The family therefore gives a controlled setting in which the influence of the coefficient lattice can be studied without changing the Alamouti encoder or its decision structure. Different values of $q$ also support different message partitions through the Chinese remainder theorem, as discussed in Section~\ref{subsec:side-information}. The quaternion algebras, their maximal orders, and the orthogonal matrix are established objects. This common structure lets us isolate the finite energy and information variance properties of the coefficient alphabets.

The main contributions are summarised as follows.

\begin{itemize}
 \item We establish exact finite energy optimality of the Eisenstein case for every $p\geq2$ in the prime indexed family, at fixed alphabet size and minimum lattice distance. Theorems~\ref{thm:finite-family-optimality} and~\ref{thm:translated-family-optimality} cover the prescribed quotient alphabets and joint optimisation of the representatives and their common translation, with an explicit positive gap for every $q>3$.

 \item We derive a controlled high-SNR expansion of the information variance and its limiting ratio to the mutual information deficit in Theorem~\ref{thm:controlled-variance-expansion}. The shared fading contribution is retained. The accompanying mutual information expansion specialises arbitrary input fading results with an explicit remainder and a coefficient determined by the complete AWGN information deficit.

 \item We use standard random-coding and second-order tools to study the reliability implications for independently faded Alamouti blocks. Numerical comparisons illustrate the finite-alphabet effects while retaining symbol-wise ML detection. Equal rate, equal energy comparisons with Golden and Silver codes place these effects in the context of symbol rate and decoding effort.
\end{itemize}

The remainder of the paper is organised as follows. Section~\ref{sec:preliminaries} recalls the required facts about quaternion algebras, maximal orders, and lattices. Section~\ref{sec:stbc} presents the channel models and the quaternionic Alamouti family. Section~\ref{sec:shaping-mi} studies shaping, finite-constellation energy, and constellation constrained mutual information. Section~\ref{sec:finite-blocklength} gives the finite-blocklength analysis. Section~\ref{sec:numerical-results} presents the numerical energy, mutual information, error-rate, short-packet, and standard-STBC comparisons. Section~\ref{sec:conclusion} concludes the paper.

\subsection*{Notation}
The sets $\mathbb Z$, $\mathbb Q$, $\mathbb R$, and $\mathbb C$ denote integers, rationals, reals, and complex numbers, respectively, with $\mathrm i^2=-1$. Quaternion orders and Euclidean lattices are denoted by $\Gamma$ and $\Lambda$. In particular, $\Gamma_q$ has coefficient lattice $\Lambda_q$ and matrix lattice $\mathcal L_q=\Phi_q(\Gamma_q)$. The sets $\mathcal X_{\Lambda,p}$, $\mathcal A_{\Lambda,p}$, and $\mathcal C_{q,p}$ denote the unnormalised quotient alphabet, its energy-normalised version, and the corresponding space--time codebook. Bold capitals denote matrices, including $\mathbf X_b$ and $\mathbf H_b$. We use $(\cdot)^{\mathsf T}$, $(\cdot)^\dagger$, and $\|\cdot\|_{\mathrm F}$ for transpose, Hermitian transpose, and Frobenius norm; $|x|$ and $|\mathcal A|$ denote scalar modulus and set cardinality. We write $h_b=\|\mathbf H_b\|_{\mathrm F}^2$ for the effective fading gain and $G(\Lambda)$ for the normalised second moment. The distributions $\mathcal N_{\mathbb C}(0,\sigma^2)$ and $\Gamma(k,\theta)$ denote circularly symmetric complex Gaussian and Gamma distributions, the latter with shape $k$ and scale $\theta$. Information quantities are in bits per complex channel use unless stated otherwise.

\section{Preliminaries}
\label{sec:preliminaries}

We recall the algebraic and geometric concepts needed for the construction. The reduced norm links the quaternion algebra to the determinant of a space--time matrix. The shaping analysis concerns the two-dimensional lattice obtained from the coefficient ring.

\subsection{Quaternion algebras and maximal orders}
\label{subsec:quaternion-orders}

We use the standard definitions of quaternion algebras and orders from \cite{berhuy2013central,voight2021quaternion}.

For the proposed construction, it is useful to write the algebra in a form that shows a quadratic subfield.  Let $K/\mathbb Q$ be quadratic, denote its nontrivial automorphism by $x\mapsto\bar{x}$, and fix $c\in\mathbb Q^{\times}$.  The crossed-product presentation
\begin{equation*}
 B=(K,c)=K\oplus uK,\quad u^2=c,\quad ux=\bar{x}u \label{eq:prelim-crossed-product}
\end{equation*}
is a quaternion algebra over $\mathbb Q$.  If $K=\mathbb Q(\sqrt d)$, then $(K,c)$ is isomorphic to $\left(\frac{d,c}{\mathbb Q}\right)$.  Direct multiplication gives
\begin{equation}
 (x+uy)(x'+uy') = xx'+c\bar y y'+u\bigl(\bar x y'+yx'\bigr), \label{eq:prelim-crossed-product-multiplication}
\end{equation}
and therefore
\begin{equation}
\operatorname{Nrd}_B(x+uy) = \operatorname{N}_{K/\mathbb Q}(x) -c\operatorname{N}_{K/\mathbb Q}(y).
\label{eq:prelim-crossed-product-norm}
\end{equation}

When $K$ is imaginary quadratic and $c<0$, the right-hand side of \eqref{eq:prelim-crossed-product-norm} is positive for every nonzero element. Therefore, $B$ is a definite quaternion algebra.

We next recall the facts about orders needed for the proposed family.  An order $\Gamma\subset B$ is a subring containing $1$ and a full $\mathbb Z$-lattice in $B$.  It is maximal if it is not properly contained in another order.  For $c\in\mathbb Z\setminus\{0\}$, the ring of integers $\mathcal O_K$ defines the natural crossed-product order
\begin{equation*}
\Gamma(K,c)=\mathcal O_K\oplus u\mathcal O_K. \label{eq:prelim-natural-order}
\end{equation*}
Equation~\eqref{eq:prelim-crossed-product-multiplication} shows that this set is closed under multiplication.

Maximality is an arithmetic property and does not by itself imply better communication performance. For the orders used here, maximality follows from the classical construction in \cite[Ex.~15.5.7]{voight2021quaternion}.

\subsection{Lattices and finite constellations}
\label{subsec:lattices-r2}

We now introduce the geometric measures needed for the shaping analysis. A lattice in $\mathbb R^2$ is written as $\Lambda=\mathbf M_{\Lambda}\mathbb Z^2 \subset\mathbb R^2\simeq\mathbb C,$ where $\mathbf M_{\Lambda}\in\mathbb R^{2\times2}$ is invertible. Its fundamental volume is $\operatorname{vol}(\Lambda)=|\det\mathbf M_{\Lambda}|$, and its minimum distance is
\begin{equation*}
 d_{\min}(\Lambda) =\min_{\lambda\in\Lambda\setminus\{0\}}|\lambda|.
 \label{eq:prelim-minimum-distance}
\end{equation*}
The Voronoi region of $\Lambda$ is
\begin{equation*}
 \mathcal V(\Lambda) =\bigl\{z\in\mathbb C:|z|\leq |z-\lambda| \text{ for all }\lambda\in\Lambda\bigr\}.
 \label{eq:prelim-voronoi-region}
\end{equation*}
It consists of the points at least as close to the origin as to any other lattice point.

The normalised second moment in dimension 2, which is unchanged by scaling, is
\begin{equation*}
 G(\Lambda) =\frac{1}{2\operatorname{vol}(\Lambda)^2} \int_{\mathcal V(\Lambda)}|z|^2\,\mathrm dz.
 \label{eq:prelim-normalised-second-moment}
\end{equation*}

For a fixed volume, this quantity measures the average energy of a large Voronoi shaped constellation \cite{Con2013}.

To define a finite alphabet clearly at the boundary, let $\mathcal V_0(\Lambda)$ be a half-open version of $\mathcal V(\Lambda)$, with boundary ties assigned in a fixed way.  For an integer $p\geq2$, define
\begin{equation*}
 \mathcal X_{\Lambda,p} =\Lambda\cap p\mathcal V_0(\Lambda).
 \label{eq:prelim-finite-alphabet}
\end{equation*}

This set contains one representative of each coset of $\Lambda/p\Lambda$ and therefore $|\mathcal X_{\Lambda,p}|=p^2$.  For uniform signalling, its average symbol energy is
\begin{equation*}
 E_s(\Lambda,p) =\frac{1}{p^2} \sum_{x\in\mathcal X_{\Lambda,p}}|x|^2.
 \label{eq:prelim-finite-energy}
\end{equation*}

We use the energy normalised alphabet $\mathcal A_{\Lambda,p} = E_s(\Lambda,p)^{-1/2}\mathcal X_{\Lambda,p},$ for which a uniform symbol has average energy one. 

For the family studied in this work, let $K_q=\mathbb Q(\sqrt{-q})$, where $q\equiv3\pmod4$ is prime, and let $\mathcal O_q=\mathbb Z[\omega_q]$ with $ \omega_q=\frac{1+\sqrt{-q}}{2}.$ Choose the embedding $\sigma_q(\sqrt{-q})=\mathrm i\sqrt q$.  The coefficient lattice $\Lambda_q=\sigma_q(\mathcal O_q)$ has generator matrix and volume
\begin{equation*}
 \mathbf M_q=
 \begin{pmatrix}
  1&1/2\\
  0&\sqrt q/2
 \end{pmatrix},\quad \operatorname{vol}(\Lambda_q)=\frac{\sqrt q}{2}.
 \label{eq:prelim-family-lattice}
\end{equation*}

For $q=3$, $\Lambda_q$ is the hexagonal lattice associated with the Eisenstein integers $\Z[\frac{1+\sqrt{-3}}{2}]$, while $\mathbb Z[\mathrm i]$ is the square lattice used by the classical Alamouti code~\cite{souza2026AESTBC}.

\section{Space--Time Block Codes}
\label{sec:stbc}

Now we state the channel and fading models used in the paper and construct the proposed family of quaternionic Alamouti codes describing its diversity and decoding properties. Since every code matrix uses two transmit antennas and two time slots, we specialise the notation to this setting from the beginning.

\subsection{Channel and fading models} \label{subsec:channel-fading-models}

Let $b\in\{1,\ldots,m\}$ denote an Alamouti block. The received signal is

\begin{equation}
 \mathbf Y_b =\sqrt{\frac{\rho}{2}}\,\mathbf H_b\mathbf X_b+\mathbf Z_b, \label{eq:stbc-channel}
\end{equation}
where $\mathbf X_b\in\mathbb C^{2\times2}$ is transmitted over two channel uses, $\mathbf H_b\in\mathbb C^{n_r\times2}$ is the channel matrix, and $\mathbf Z_b\in\mathbb C^{n_r\times2}$ is additive noise.  The entries of $\mathbf H_b$ and $\mathbf Z_b$ are independent and distributed as $\mathcal N_{\mathbb C}(0,1)$. The receiver knows $\mathbf H_b$, while the transmitter knows only its distribution.  Uniform selection from the modulation codebooks below satisfies

\begin{equation}
 \mathbb E\!\left[\|\mathbf X_b\|_{\mathrm F}^2\right]=4, \label{eq:stbc-energy-convention}
\end{equation}
so $\rho$ is the average receive SNR per receive antenna. For one space--time block, coherent maximum-likelihood (ML) detection gives

\begin{equation*}
 \widehat{\mathbf X}_b =\arg\min_{\mathbf X\in\mathcal C} \left\|\mathbf Y_b- \sqrt{\frac{\rho}{2}}\,\mathbf H_b\mathbf X\right\|_{\mathrm F}^2.
 \label{eq:stbc-ml-rule}
\end{equation*}

We use the \emph{block-memoryless model}: $\mathbf H_b$ is constant during block $b$ and independent across Alamouti blocks. Packet-wise quasi-static fading, in which the same channel remains fixed throughout a coded packet, is outside the finite-blocklength analysis below.

For coherent Rayleigh fading, a finite code $\mathcal C$ has full diversity if
\begin{equation*}
 \det(\mathbf X-\mathbf X')\neq0 \quad\text{for all distinct }\mathbf X,\mathbf X'\in\mathcal C.
 \label{eq:stbc-full-diversity-criterion}
\end{equation*}

For an infinite space--time lattice $\Lambda$, the minimum determinant is $\delta_{\min}(\Lambda) =\min_{\mathbf 0\neq\mathbf X\in\Lambda} |\det(\mathbf X)|^2.$ The lattice has the non-vanishing determinant (NVD) property if this value is positive \cite{tarokh1998criteria,oggier2007cyclic}.  

\subsection{A quaternionic Alamouti family}\label{subsec:maximal-quaternionic-alamouti-family}

Recall $K_q$, $\mathcal O_q$, and the lattice $\Lambda_q=\sigma_q(\mathcal O_q)$ from Section~\ref{subsec:lattices-r2}.  For every prime $q\equiv3\pmod4$, define $B_q=(K_q,-1)=K_q\oplus  uK_q, u^2=-1, ux=\bar xu,$ and 
$$\Gamma_q=\mathcal O_q\oplus u\mathcal O_q.$$
Here $\Gamma_q$ is a quaternion order, while $\Lambda_q$ is its coefficient lattice in $\mathbb R^2$. This notation keeps the arithmetic and geometric objects distinct.

\begin{Prop}\label{thm:maximal-quaternionic-alamouti-family}
Let $q\equiv3\pmod4$ be prime.  Then $B_q$ is a definite quaternion algebra and $\Gamma_q$ is a maximal order in $B_q$.  Moreover, the map
\begin{equation*}
 \Phi_q(x_0+ux_1)= \begin{pmatrix}
  \sigma_q(x_0)&-\overline{\sigma_q(x_1)}\\
  \sigma_q(x_1)& \overline{\sigma_q(x_0)}
 \end{pmatrix}
 \label{eq:stbc-family-embedding}
\end{equation*}
is an injective $\mathbb Q$-algebra representation of $B_q$ in $M_2(\mathbb C)$.  For every $z=x_0+ux_1\in B_q$,
\begin{itemize}
    \item[(i)] $\Phi_q(z)^{\dagger}\Phi_q(z) =\operatorname{Nrd}_{B_q}(z)\mathbf I_2,$
    \item[(ii)] $\det\Phi_q(z) =\operatorname{Nrd}_{B_q}(z) =|\sigma_q(x_0)|^2+|\sigma_q(x_1)|^2,$
    \item[(iii)] $\frac12\|\Phi_q(z)\|_{\mathrm F}^2 =\operatorname{Nrd}_{B_q}(z).$
\end{itemize}
Consequently, $\mathcal L_q=\Phi_q(\Gamma_q)$ is isometric to $\Lambda_q\oplus\Lambda_q$ under the normalised Frobenius metric, has full diversity, and satisfies $\delta_{\min}(\mathcal L_q)=1.$
\end{Prop}

\begin{proof}
The definiteness of $B_q$ follows from the positive norm in \eqref{eq:prelim-crossed-product-norm} with $c=-1$. The algebra and maximal order are those of \cite[Ex.~15.5.7]{voight2021quaternion}, after identifying its generators with $u$ and $\sqrt{-q}$. Thus maximality is recalled here rather than claimed as a new result. The multiplication rule \eqref{eq:prelim-crossed-product-multiplication} verifies the representation, and direct matrix multiplication proves (i)--(iii).

For $0\neq z=x_0+ux_1\in\Gamma_q$, $\operatorname{Nrd}_{B_q}(z)$ is a positive integer. Therefore $|\det\Phi_q(z)|\geq1$.  Equality is attained for $z=1$, which proves $\delta_{\min}(\mathcal L_q)=1$. The determinant of every nonzero difference is also nonzero, so the lattice has full diversity.
\end{proof}

For $p\geq2$, let $\mathcal A_{q,p}:=\mathcal A_{\Lambda_q,p}$ be the energy normalised alphabet. The proposed finite code is
\begin{equation}
 \mathcal C_{q,p} =\left\{ \begin{pmatrix}
  a_0&-\overline{a_1}\\
  a_1& \overline{a_0}
 \end{pmatrix}: a_0,a_1\in\mathcal A_{q,p} \right\}.
 \label{eq:stbc-family-codebook}
\end{equation}
It has $|\mathcal C_{q,p}|=p^4$ and transmits two complex symbols over two channel uses. Thus its rate is one complex symbol per channel use. Since a uniform symbol in $\mathcal A_{q,p}$ has average energy one, \eqref{eq:stbc-energy-convention} holds.

For two distinct codewords, write $\Delta x_k=x_k-x'_k$, $k\in\{0,1\}$.  Their difference satisfies
\begin{align}
 \Delta\mathbf X^{\dagger}\Delta\mathbf X &=\bigl(|\Delta x_0|^2+|\Delta x_1|^2\bigr)\mathbf I_2, \label{eq:stbc-finite-orthogonality}\\
 \det(\Delta\mathbf X) &=|\Delta x_0|^2+|\Delta x_1|^2>0. \nonumber
\end{align}

Hence every code in the family has diversity order $2n_r$.  Notice that the NVD statement in Proposition~\ref{thm:maximal-quaternionic-alamouti-family} applies to the infinite lattice.  The energy normalisation used for a finite alphabet rescales its determinants.

\subsection{Symbol-wise decoding} \label{subsec:stbc-decoding-reference-cases}

The orthogonality in \eqref{eq:stbc-finite-orthogonality} gives the same simple decoder for every $q$.  After Alamouti combining and noise normalisation, the two symbols satisfy

\begin{equation}
 \widetilde Y_{k,b} =\sqrt{\frac{\rho h_b}{2}}\,X_{k,b} +\widetilde Z_{k,b}, \quad k\in\{0,1\},
 \label{eq:stbc-equivalent-scalar-channel}
\end{equation}
where $h_b=\|\mathbf H_b\|_{\mathrm F}^2 \sim\Gamma(2n_r,1), \widetilde Z_{k,b}\sim\mathcal N_{\mathbb C}(0,1).$  Conditioned on $h_b$, the two noise variables are independent. Therefore ML detection reduces to

\begin{equation}
 \widehat X_{k,b} =\arg\min_{x\in\mathcal A_{q,p}} \left|\widetilde Y_{k,b} -\sqrt{\frac{\rho h_b}{2}}\,x\right|^2, \quad k\in\{0,1\}.
 \label{eq:stbc-symbolwise-decoding}
\end{equation}

Thus all cases have the same symbol-wise decoding structure. Their differences lie in the geometry of $\Lambda_q$, rather than in the number of symbols jointly searched by the receiver.

Two cases will be used as references.  When $q=3$, $\mathcal O_3=\mathbb Z[(1+\sqrt{-3})/2]$ and $\Lambda_3$ is hexagonal. This case is the Alamouti--Eisenstein code~\cite{souza2026AESTBC}.  The classical Alamouti code is obtained by using the square coefficient lattice $\mathbb Z[i]$ in \eqref{eq:stbc-family-codebook} \cite{alamouti1998}.  Its natural rational quaternion order is not maximal, so it is a reference rather than a member of the maximal family.

\subsection{Message partitions and receiver side information}
\label{subsec:side-information}

The coefficient rings of the proposed family also provide structured message partitions for broadcasting with receiver side information. Following the layered Alamouti index-coding principle in~\cite[Sec.~IX]{huang2018layered}, Chinese remainder theorem (CRT) partitions~\cite{huang2017lattice} can be applied separately to the two Alamouti symbols. Known message components then restrict each scalar alphabet to a subset, preserving symbol-wise ML detection. The corresponding minimum determinant
comparison follows from \eqref{eq:stbc-finite-orthogonality}.

For $p=\ell$ prime, suppose that $f_q(t)=t^2-t+(q+1)/4$ has two distinct roots in $\mathbb F_\ell$. The CRT gives
\[
 \mathcal O_q/\ell\mathcal O_q
 \simeq \mathbb F_\ell\times\mathbb F_\ell,
\]
and evaluation at the roots labels two $\ell$-ary message components in each symbol. For example, the choices $(q,\ell)=(7,2),(11,3),(19,5)$ give $f_q(t)=t(t-1)$ modulo $\ell$, supporting binary, ternary, and $5$-ary component pairs, respectively. 

A comparison with both Gaussian and Eisenstein alphabets is obtained at $\ell=11$. For $q=7$, the factorisation $f_7(t)=(t-5)(t-7)$ modulo $11$ gives two $11$-ary CRT components, whereas both $\mathbb Z[\mathrm i]/11\mathbb Z[\mathrm i]$ and $\mathcal O_3/11\mathcal O_3$ are isomorphic to $\mathbb F_{121}$. All three scalar alphabets have $121$ points and permit two-component labels. The distinction is that the CRT labels preserve componentwise addition and multiplication. Thus, $q=7$ supports a ring-compatible product partition unavailable from either conventional ring at this modulus.

This algebraic flexibility is related to the lattice constructions over imaginary quadratic rings in~\cite{huang2018lattices}, which connect lattice operations with finite field coding for compute-and-forward. Extending the present Alamouti family to that network setting is beyond the scope of this work.

For an illustration with $p=2$, choose the unit average energy alphabet
\[
 \sqrt{\frac{8}{q+3}}\, \{0,1,\omega_q,\omega_q-1\},
\]
with labels $00,11,01,10$, respectively. These are CRT labels when $q\equiv7\pmod8$. Knowing either bit leaves two candidate points. The minimum squared distance of the full alphabet and the worst conditional minimum squared distance, over either known bit and its value, are
\[
 d_{\min}^2=\frac{8}{q+3},
 \qquad
 d_{\mathrm{SI}}^2=\frac{2(q+1)}{q+3}.
\]
Increasing $q$ decreases the first distance and increases the second. Figure~\ref{fig:side-information-constellations} illustrates this trade-off for $q=3$ and $q=7$. The numerical comparisons below use the point-to-point setting without side information.

\begin{figure}[t]
\centering
\resizebox{\columnwidth}{!}{%
\begin{tikzpicture}[
  x=9.3mm,y=9.3mm,text=black,draw=black,
  font=\fontsize{8}{9}\selectfont,
  kept/.style={
    circle,fill=black,draw=black,
    inner sep=0pt,minimum size=4.2pt},
  removed/.style={
    circle,fill=white,draw=black,line width=.65pt,
    inner sep=0pt,minimum size=4.2pt},
  bit/.style={inner sep=1.5pt}
]
\path[use as bounding box]
  (-2.15,-2.78) rectangle (6.90,1.97);

\foreach \q/\offset/\panel/\distance in
  {3/0/a/{4/3},7/4.75/b/{8/5}} {
  \begin{scope}[shift={(\offset,0)}]

    \pgfmathsetmacro{\s}{sqrt(8/(\q+3))}
    \pgfmathsetmacro{\u}{\s/2}
    \pgfmathsetmacro{\h}{\s*sqrt(\q)/2}
    \pgfmathsetmacro{\v}{\s*(\q-1)/(4*sqrt(\q))}
    \pgfmathsetmacro{\w}{\s*(\q+1)/(4*sqrt(\q))}

    \def\cellpath{
      (\u,\v)--(0,\w)--(-\u,\v)--
      (-\u,-\v)--(0,-\w)--(\u,-\v)--cycle}

    \begin{scope}
      \clip (-2.10,-1.60) rectangle (2.10,1.60);

      \begin{scope}[scale=2]
        \fill[black!10] \cellpath;
      \end{scope}

      \foreach \m in {-3,...,3} {
        \foreach \n in {-2,...,2} {
          \begin{scope}[
            shift={({\s*(\m+\n/2)},{\n*\h})}]
            \draw[black!22,line width=.3pt] \cellpath;
          \end{scope}
        }
      }

      \begin{scope}[scale=2]
        \draw[black!70,line width=.65pt] \cellpath;
      \end{scope}

      \draw[->,black!60,line width=.35pt]
        (-2.06,0)--(2.06,0);
      \draw[->,black!60,line width=.35pt]
        (0,-1.55)--(0,1.56);

      \foreach \m in {-3,...,3} {
        \foreach \n in {-2,...,2} {
          \fill[black!50]
            ({\s*(\m+\n/2)},{\n*\h}) circle (.8pt);
        }
      }
    \end{scope}

    \node[anchor=north east,inner sep=1pt]
      at (2.10,-.06) {$\operatorname{Re}$};
    \node[anchor=south,inner sep=1pt]
      at (0,1.64) {$\operatorname{Im}$};

    \draw[line width=.95pt] (0,0)--(\u,\h);

    \node[kept] at (0,0) {};
    \node[kept] at (\u,\h) {};
    \node[removed] at (\s,0) {};
    \node[removed] at (-\u,\h) {};

    \node[bit,below left=1pt] at (0,0) {$00$};
    \node[bit,above=2pt] at (\s,0) {$11$};
    \node[bit,above=2pt] at (\u,\h) {$01$};
    \node[bit,above=2pt] at (-\u,\h) {$10$};

    \node at (0,-2.04)
      {(\panel) $q=\q,\quad
        d_{\mathrm{SI}}^2=\distance$};
  \end{scope}
}
\end{tikzpicture}%
}

\vspace{-0.3cm}
\caption{Voronoi shaping regions of $2\Lambda_q$, scaled
with the four-point alphabets to unit average energy.
Large markers show the four selected representatives,
with boundary ties resolved as displayed. Filled large
markers remain when $b_1=0$, with squared separation
$d_{\mathrm{SI}}^2$. Both panels use identical axis scales.
The labeling is a CRT labeling only for $q=7$.}
\label{fig:side-information-constellations}
\end{figure}
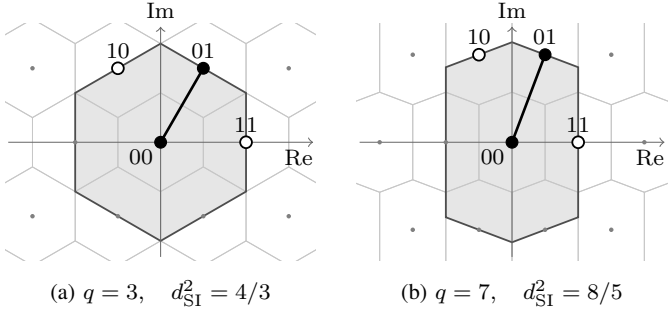

\section{Shaping and Mutual Information Analysis} \label{sec:shaping-mi}

All members of the family retain the same Alamouti structure, diversity order, and symbol-wise decoder. We now isolate the effect of the coefficient lattice $\Lambda_q$. The Voronoi calculation provides the geometric background for the exact finite energy comparisons in Theorems~\ref{thm:finite-family-optimality} and~\ref{thm:translated-family-optimality}. We then study constellation-constrained mutual information and the pairwise error probability (PEP) implication. All comparisons use the energy convention in \eqref{eq:stbc-energy-convention}.

\subsection{Shaping and finite constellation energy} \label{subsec:shaping-energy}

Recall that $\Lambda_q$ is generated by $(1,0)$ and $(1/2,\sqrt q/2)$. Thus, for $m,n\in\mathbb Z$,
\begin{equation}
 \left|m+n\omega_q\right|^2 =m^2+mn+\frac{q+1}{4}n^2.
 \label{eq:performance-quadratic-form}
\end{equation}
In particular, $d_{\min}(\Lambda_q)=1$. The next result gives the shaping quantities of every member of the family.

\begin{Teo}\label{thm:family-voronoi-energy}
Let $q\equiv3\pmod4$ be prime, and let $\mathcal V_q=\mathcal V(\Lambda_q)$.  Then
\begin{equation}
 \mathcal V_q= \left\{(x,y)\in\mathbb R^2: |x|\leq\frac12,\quad |x|+\sqrt q\,|y|\leq\frac{q+1}{4} \right\}.
 \label{eq:family-voronoi-cell}
\end{equation}
We have, (i) $\operatorname{vol}(\Lambda_q) =\frac{\sqrt q}{2},$ (ii) $\frac{1}{\operatorname{vol}(\Lambda_q)} \int_{\mathcal V_q}|z|^2\,\mathrm dz =\frac{q^2+4q-1}{48q},$ (iii) $G(\Lambda_q)=\frac{q^2+4q-1}{48q^{3/2}}.$ For fixed $q$ and for $p\rightarrow\infty$, the finite alphabet $\mathcal X_{\Lambda,p}$ satisfies
\begin{equation}
 E_s(\Lambda_q,p) =p^2\frac{q^2+4q-1}{48q}+O_q(p), \label{eq:family-finite-energy}
\end{equation}
\end{Teo}

\begin{proof}
Put $b_1=(1,0)$ and $b_2=(1/2,\sqrt q/2)$.  The perpendicular bisectors associated with $\pm b_1, \pm b_2,$ and $ \pm(b_2-b_1)$ define the polygon on the right-hand side of \eqref{eq:family-voronoi-cell}.  Figure~\ref{fig:family-voronoi-cells} illustrates these six bisectors for $q=3$ and $q=7$.  The two panels use the same scale, showing the regular Eisenstein cell and the elongation that occurs for $q>3$.  Notice that, for $q>3$, the four oblique vectors are Voronoi-relevant, although they need not have minimum length. Since the associated half-spaces are necessary Voronoi inequalities, $\mathcal V_q$ is contained in their intersection, denoted by $P_q$.

For $|x|\leq1/2$, define $h_q(x)=\frac{(q+1)/4-|x|}{\sqrt q}.$ Then $P_q$ is described by $|y|\leq h_q(x)$.  Direct integration gives
\begin{align*}
 \int_{-1/2}^{1/2}2h_q(x)\,\mathrm dx &=\frac{\sqrt q}{2},\\
 \int_{-1/2}^{1/2} \left(2x^2h_q(x)+\frac{2}{3}h_q(x)^3\right)\mathrm dx &=\frac{q^2+4q-1}{96\sqrt q}.
\end{align*}
The first value equals $|\det\mathbf M_q|$, which is also the area of a Voronoi region of $\Lambda_q$. The containment and the equality of areas prove (i). Dividing the second expression by the first proves (ii).  Equation (iii) follows from the definition of $G(\Lambda)$. Finally, $p^{-1}\mathcal X_{\Lambda_q,p}$ is a Riemann sum on $\mathcal V_q$ for the scaled lattice $p^{-1}\Lambda_q$. The boundary of $\mathcal V_q$ is polygonal, so the boundary contribution is $O_q(p)$
after returning to the original scale. This proves \eqref{eq:family-finite-energy}.
\end{proof}

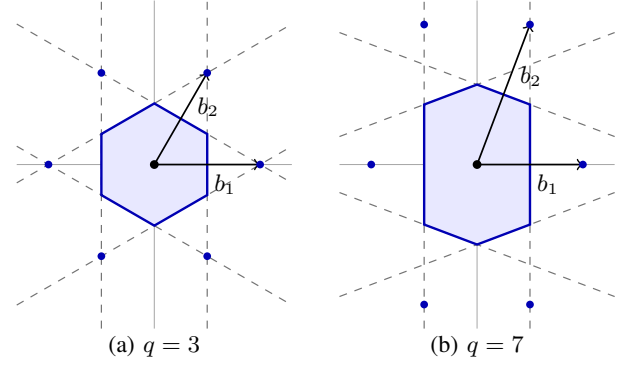
\begin{figure}[t]
 \centering
 \begin{tikzpicture}[x=1.4cm, y=1.4cm,
   axis/.style={draw=black!35, line width=0.35pt},
   bisector/.style={draw=black!55, dashed, line width=0.45pt},
   cell edge/.style={draw=blue!70!black, line width=0.9pt},
   basis/.style={draw=black, ->, line width=0.65pt},
   every node/.style={font=\small}
 ]

 \newcommand{\drawvoronoipanel}[3]{%
  \begin{scope}[shift={(#1,0)}]
   \pgfmathsetmacro{\sq}{sqrt(#2)}
   \pgfmathsetmacro{\cval}{(#2+1)/4}
   \pgfmathsetmacro{\ytop}{\cval/\sq}
   \pgfmathsetmacro{\yside}{(\cval-0.5)/\sq}
   \pgfmathsetmacro{\ynbr}{\sq/2}

   \begin{scope}
    \clip (-1.30,-1.55) rectangle (1.30,1.55);

    \draw[axis] (-1.30,0)--(1.30,0);
    \draw[axis] (0,-1.55)--(0,1.55);

    \fill[blue!9]
      (0,\ytop)--(0.5,\yside)--(0.5,-\yside)--
      (0,-\ytop)--(-0.5,-\yside)--(-0.5,\yside)--cycle;

    \draw[bisector] (-0.5,-1.55)--(-0.5,1.55);
    \draw[bisector] ( 0.5,-1.55)--( 0.5,1.55);
    \draw[bisector]
      (-1.30,{(\cval+1.30)/\sq})--(1.30,{(\cval-1.30)/\sq});
    \draw[bisector]
      (-1.30,{(-\cval+1.30)/\sq})--(1.30,{(-\cval-1.30)/\sq});
    \draw[bisector]
      (-1.30,{(\cval-1.30)/\sq})--(1.30,{(\cval+1.30)/\sq});
    \draw[bisector]
      (-1.30,{(-\cval-1.30)/\sq})--(1.30,{(-\cval+1.30)/\sq});

    \draw[cell edge]
      (0,\ytop)--(0.5,\yside)--(0.5,-\yside)--
      (0,-\ytop)--(-0.5,-\yside)--(-0.5,\yside)--cycle;

    \draw[basis] (0,0)--(1,0)
      node[pos=0.67,below] {$b_1$};
    \draw[basis] (0,0)--(0.5,\ynbr)
      node[pos=0.63,right] {$b_2$};

    \fill[blue!70!black] ( 1,0) circle (1.45pt);
    \fill[blue!70!black] (-1,0) circle (1.45pt);
    \fill[blue!70!black] ( 0.5, \ynbr) circle (1.45pt);
    \fill[blue!70!black] (-0.5, \ynbr) circle (1.45pt);
    \fill[blue!70!black] ( 0.5,-\ynbr) circle (1.45pt);
    \fill[blue!70!black] (-0.5,-\ynbr) circle (1.45pt);
    \fill[black] (0,0) circle (1.65pt);
   \end{scope}

   \node at (0,-1.72) {#3};
  \end{scope}
 }

 \drawvoronoipanel{0}{3}{(a) $q=3$}
 \drawvoronoipanel{3.05}{7}{(b) $q=7$}
 \end{tikzpicture}
 \caption{Voronoi geometry of the coefficient lattices at unit minimum
 distance.  The blue points are the six Voronoi-relevant vectors
 $\pm b_1$, $\pm b_2$, and $\pm(b_2-b_1)$, and the dashed lines are
 their perpendicular bisectors.  The corresponding half-spaces that
 contain the origin form $\mathcal V_q$.  The Eisenstein cell for $q=3$
 is regular, whereas the $q=7$ cell is elongated.  Both panels use the
 same Euclidean scale.}
 \label{fig:family-voronoi-cells}
\end{figure}

Theorem~\ref{thm:family-voronoi-energy} produces two different comparisons. The conventional shaping gain compares normalised second moments and is therefore an equal volume comparison. Relative to the square lattice, it is
\begin{equation*}
10\log_{10}\frac{G(\mathbb Z[\mathrm i])}{G(\Lambda_q)} =10\log_{10}\frac{4q^{3/2}}{q^2+4q-1} \quad\text{dB}, \label{eq:conventional-shaping-gain}
\end{equation*}
where $G(\mathbb Z[\mathrm i])=1/12$.  In contrast, using the same cardinality $p^2$ and the same minimum distance gives the asymptotic energy advantage
\begin{equation*}
 \gamma_E(q) =10\log_{10}\frac{E_s(\mathbb Z[\mathrm i],p)}{E_s(\Lambda_q,p)} \rightarrow 10\log_{10}\frac{8q}{q^2+4q-1} \quad\text{dB}.
 \label{eq:fixed-index-energy-gain}
\end{equation*}

Here we used $E_s(\mathbb Z[\mathrm i],p)=p^2/6+O(p)$. The two gains are related by
\begin{equation*}
 \gamma_E(q)=10\log_{10}\frac{G(\mathbb Z[\mathrm i])}{G(\Lambda_q)} +10\log_{10}\frac{2}{\sqrt q}.
 \label{eq:relation-between-shaping-gains}
\end{equation*}

The second term in the sum above is the fundamental volume effect at fixed minimum distance. For $q=3$, the $0.792$ dB energy advantage is the sum of the $0.167$ dB equal volume shaping gain and a $0.625$ dB fundamental volume advantage.

\begin{Cor} \label{cor:eisenstein-family-optimality}
Among all primes $q\equiv3\pmod4$, the Eisenstein case $q=3$ uniquely minimises both $G(\Lambda_q)$ and the leading term of $E_s(\Lambda_q,p)$.  Its shaping gain over the square lattice is $10\log_{10}\frac{3\sqrt3}{5}\simeq0.167\ \text{dB}.$ At equal cardinality and equal minimum distance, its asymptotic energy advantage is $10\log_{10}\frac65\simeq0.792\ \text{dB}.$
\end{Cor}

\begin{proof}
For real $q\geq3$, differentiation of (iii) gives
\(
 \frac{\mathrm d}{\mathrm dq}G(\Lambda_q) =\frac{(q-1)(q-3)}{96q^{5/2}}\geq0.
\)
The coefficient in \eqref{eq:family-finite-energy} can be written as $(q+4-q^{-1})/48$, which is strictly increasing for $q>0$. Consequently, both quantities are minimised at $q=3$.  Substitution of $q=3$ gives $10\log_{10}\frac{3\sqrt3}{5}\simeq0.167\ \text{dB}$ and $10\log_{10}\frac65\simeq0.792\ \text{dB}$.  Finally,
\[
 \frac{8q}{q^2+4q-1}>1 \quad\Longleftrightarrow\quad q<2+\sqrt5.
\]
Only $q=3$ satisfies this inequality in the proposed family. Similarly, $10\log_{10}\frac{G(\mathbb Z[\mathrm i])}{G(\Lambda_q)}$ is positive exactly when $q<(1+\sqrt2)^2$, which again leaves only $q=3$.
\end{proof}

Only the Eisenstein case has positive shaping and asymptotic energy gains over the square lattice. The following gives the exact correction for finite square and hexagonal alphabets.

\begin{Prop}\label{prop:exact-finite-eisenstein-energy}
For every integer $p\geq2$, define
\begin{equation*}
 \chi_3(p)=
 \begin{cases}
  0,&3\mid p,\\
  1,&p\equiv1\pmod3,\\
 -1,&p\equiv2\pmod3.
 \end{cases}
\end{equation*}
Then 
\begin{align*}
E_s(\mathbb Z[\mathrm i],p) &= \begin{cases}
  (p^2-1)/6,&p\text{ odd},\\
  (p^2+2)/6,&p\text{ even},
 \end{cases} \text{ and }\\
E_s(\Lambda_3,p) &=\dfrac{1}{36} \left(5p^2+3-\frac{8\chi_3(p)}{p}\right).
\end{align*}
For odd $p$, the exact finite energy ratio is $\frac{E_s(\mathbb Z[\mathrm i],p)}{E_s(\Lambda_3,p)}=\frac{6(p^2-1)}{5p^2+3-8\chi_3(p)/p},$ which converges to $6/5$.
\end{Prop}

\begin{proof}
For the square lattice, take a half-open set of consecutive integer representatives in each real dimension.  The usual formula for the sum of squares gives $E_s(\mathbb Z[\mathrm i],p)$.  The two cases occur because an even number of integer representatives cannot be symmetric about zero. For $\Lambda_3$, write a point as $m+n\omega_3$ and set $a=m-n$ and $b=m+2n$.  A half-open version of $p\mathcal V_3$ is described by
\begin{align*}
 -p<a\leq p,\qquad &-p<b\leq p,\\
 -p<a+b\leq p,\qquad &a\equiv b\pmod3.
\end{align*}
It contains $p^2$ points, and $|m+n\omega_3|^2=\frac{a^2+ab+b^2}{3}.$

Writing $b=a+3r$ and applying the elementary formulas for sums of integers and squares gives
\[
 \sum_{x\in\mathcal X_{\Lambda_3,p}}|x|^2 =\frac{5p^4+3p^2-8p\chi_3(p)}{36}.
\]
Division by $p^2$ proves $E_s(\Lambda_3,p)$. Points assigned differently on opposite boundary faces have the same norm, so the result does not depend on the fixed tie rule.
\end{proof}

The preceding formulas allow an exact comparison for every finite alphabet size, at fixed cardinality $p^2$ and minimum lattice distance one. Define
\begin{equation*}
 \xi_p=\frac1p\sum_{b=0}^{p-1} \min_{k\in\mathbb Z}|b+kp|^2 =
 \begin{cases}
  (p^2-1)/12,&p\text{ odd},\\
  (p^2+2)/12,&p\text{ even}.
 \end{cases}
 \label{eq:centred-residue-energy}
\end{equation*}

\begin{Teo}\label{thm:finite-family-optimality}
For every admissible $q\geq3$ and every integer $p\geq2$,
\begin{equation}
 E_s(\Lambda_q,p) \geq E_s(\Lambda_3,p)+\frac{q-3}{4}\xi_p.
 \label{eq:finite-family-energy-bound}
\end{equation}
Consequently, $q=3$ is the unique finite-energy minimiser in the proposed family at fixed cardinality $p^2$ and fixed lattice minimum distance one.  Moreover, for every admissible $q\geq7$,
\begin{equation}
 E_s(\Lambda_3,p)\leq E_s(\mathbb Z[\mathrm i],p)<E_s(\Lambda_q,p),
\label{eq:finite-family-square-ordering}
\end{equation}
where the first inequality is an equality only for $p=3$.
\end{Teo}

\begin{proof}
In lattice coordinates, \eqref{eq:performance-quadratic-form} gives
\begin{equation*}
 m^2+mn+\frac{q+1}{4}n^2
 =m^2+mn+n^2+\frac{q-3}{4}n^2.
\end{equation*}
The Voronoi alphabet selects a minimum norm representative of every coset $(a,b)+p\mathbb Z^2$.  Every representative of this coset has $n\equiv b\pmod p$, and hence $n^2\geq\min_{k\in\mathbb Z}|b+kp|^2$.  Minimising the quadratic form in each coset, summing over all $p^2$ cosets, and dividing by $p^2$ gives \eqref{eq:finite-family-energy-bound}.

For $q\geq7$, this bound gives $E_s(\Lambda_q,p)\geq E_s(\Lambda_3,p)+\xi_p$. The square-lattice formula gives $E_s(\mathbb Z[\mathrm i],p)=2\xi_p$, while direct substitution in $E_s(\Lambda_3,p)$ shows that $E_s(\Lambda_3,p)>\xi_p$.  This proves the strict second inequality in \eqref{eq:finite-family-square-ordering}. Comparing $E_s(\mathbb Z[\mathrm i],p)$ and $E_s(\Lambda_3,p)$ proves the first inequality and its equality condition.
\end{proof}

For completeness, the classical packing and quantisation properties of the hexagonal lattice give the following consequence \cite{Con2013}.

\begin{Cor}\label{thm:broad-eisenstein-optimality}
Let $\Lambda\subset\mathbb R^2$ be any lattice with $d_{\min}(\Lambda)=1$, and define
\begin{equation*}
 \mathcal E(\Lambda) =\frac{1}{\operatorname{vol}(\Lambda)} \int_{\mathcal V(\Lambda)}|z|^2\,\mathrm dz.
 \label{eq:coefficient-lattice-energy}
\end{equation*}
Then $\mathcal E(\Lambda)\geq\frac{5}{36},$ with equality if and only if $\Lambda$ is isometric to the unit minimum distance hexagonal lattice.  Moreover, $E_s(\Lambda,p)=p^2\mathcal E(\Lambda)+O_{\Lambda}(p).$ Consequently, $\liminf_{p\to\infty}\frac{E_s(\Lambda,p)}{p^2} \geq\frac{5}{36}.$ 
\end{Cor}

\begin{proof}
Let $A=\operatorname{vol}(\Lambda)$.  Discs of radius $1/2$ centred at the lattice points have disjoint interiors.  The optimal lattice packing bound in $\R^2$ gives $A\geq\sqrt3/2$, with equality only for the hexagonal lattice.  The optimal two-dimensional lattice quantiser satisfies
\[
 G(\Lambda)\geq G(\Lambda_3)=\frac{5}{36\sqrt3},
\]
again with equality only for the hexagonal lattice.  From the definition of $G(\Lambda)$,
\[
 \mathcal E(\Lambda)=2A G(\Lambda).
\]
Combining the two bounds gives $\mathcal E(\Lambda)\geq\frac{5}{36}$. The Riemann-sum argument used in the proof of Theorem~\ref{thm:family-voronoi-energy} gives the results.
\end{proof}

Corollary~\ref{thm:broad-eisenstein-optimality} gives asymptotic energy optimality among Alamouti constructions with a two-dimensional coefficient lattice and separate Voronoi shaping of the two complex symbols.

We next allow each member of the family to optimise both its complete quotient representatives and their common translation. This tests whether the energy ordering depends on the prescribed representatives or their choice of origin.

\begin{Lema}\label{lem:uniform-residue-variance}
Let $N$ be an integer valued random variable with finite second moment. If $N\bmod p$ is uniform on $\{0,\ldots,p-1\}$, then $\operatorname{Var}(N)\geq\frac{p^2-1}{12}.$
\end{Lema}
\begin{proof}
For any $t\in\mathbb R$, conditioning on the residue gives
\[
 \mathbb E[(N-t)^2]\geq
 \frac1p\sum_{b=0}^{p-1}\min_{k\in\mathbb Z}(b+kp-t)^2.
\]
A half-open interval of length $p$ centred at $t$ contains $p$ consecutive integers, one in each residue class. Each is a nearest representative to $t$ in its class; an endpoint tie has the same squared distance. Denote their average by $\bar a$. Their mean squared distance from $t$ equals
\[
 \frac{p^2-1}{12}+(\bar a-t)^2\geq\frac{p^2-1}{12}.
\]
Taking $t=\mathbb E[N]$ proves the claim.
\end{proof}

Let $\mathscr R_p$ be the collection of complete representative sets for $\mathbb Z^2/p\mathbb Z^2$. Define
\begin{equation}
 E_{q,p}^{\star}
 =\inf_{\substack{S\in\mathscr R_p\\c\in\mathbb C}}
 \frac1{p^2}\sum_{(m,n)\in S}|m+n\omega_q-c|^2.
 \label{eq:translated-representative-energy}
\end{equation}
Here $\omega_q$ is identified with its complex embedding. The infimum permits a different representative set and translation for each $q$.

\begin{Teo}\label{thm:translated-family-optimality}
For every admissible $q\geq3$ and every integer $p\geq2$,
\begin{equation}
 E_{q,p}^{\star}\geq E_{3,p}^{\star}
       +\frac{(q-3)(p^2-1)}{48}.
 \label{eq:translated-family-energy-bound}
\end{equation}
Thus $q=3$ uniquely minimises the optimised energy in the proposed family.
\end{Teo}
\begin{proof}
Fix $S\in\mathscr R_p$, and let $(U,V)$ be uniform on $S$. The minimising translation is $c=\mathbb E[U+V\omega_q]$. Since $\omega_q=(1+\mathrm i\sqrt q)/2$, the resulting energy is
{\small \begin{align*}
 F_q(S)
 &=\operatorname{Var}(U+V/2)+\frac q4\operatorname{Var}(V)=F_3(S)+\frac{q-3}{4}\operatorname{Var}(V).
\end{align*}}
Every residue of $V$ occurs exactly $p$ times in $S$. Lemma~\ref{lem:uniform-residue-variance} therefore gives
\[
 F_q(S)\geq E_{3,p}^{\star}+\frac{(q-3)(p^2-1)}{48}.
\]
Taking the infimum over $S$ proves \eqref{eq:translated-family-energy-bound}. Its positive gap for $q>3$ proves uniqueness.
\end{proof}

Classical hexagonal quantisation and densest packing explain the normalised second moment result and Corollary~\ref{thm:broad-eisenstein-optimality}~\cite{Con2013}. Theorems~\ref{thm:finite-family-optimality} and~\ref{thm:translated-family-optimality} give exact finite alphabet comparisons, including the case where every $q$ may choose its own representatives and translation. Their positive gaps establish uniqueness in the prime indexed family. These energy results do not establish mutual information or variance optimality: the high-SNR coefficients depend on the complete AWGN information functions.

\subsection{Constellation-constrained mutual information} \label{subsec:constellation-constrained-mi}

We now use the equivalent scalar channel in \eqref{eq:stbc-equivalent-scalar-channel}. To simplify the notation, drop the block index and write

\begin{equation*}
 \widetilde Y=\sqrt{\frac{\rho h}{2}}X+Z, \qquad h\sim\Gamma(2n_r,1),
 \label{eq:mi-scalar-channel}
\end{equation*}
where $X$ is uniform on the energy-normalised alphabet $\mathcal A_{q,p}$, $Z\sim\mathcal N_{\mathbb C}(0,1)$, and the receiver knows $h$.  Put $M=p^2$.  Substitution of $\widetilde Y=\sqrt{\rho h/2}\,X+Z$ in the likelihood gives the exact mutual information.  For compactness, define
{\small 
\begin{equation}
 S_{\rho}(x;h,Z)= \sum_{x'\in\mathcal A_{q,p}} \exp\!\left(-\left|Z+\sqrt{\frac{\rho h}{2}}(x-x')\right|^2+|Z|^2\right). \label{eq:mi-log-sum}
\end{equation}}

Then
{\small \begin{align}
 I_{q,p}(\rho) &=I(X;\widetilde Y\mid h)=\log_2M-\mathbb E_{X,h,Z}\left[\log_2 S_{\rho}(X;h,Z)\right].
 \label{eq:exact-constellation-mi}
\end{align}}

The exponential in \eqref{eq:mi-log-sum} is natural, while the logarithm is base two.  Thus $I_{q,p}(\rho)$ is measured in bits per complex symbol. One Alamouti block carries two such symbols in two channel uses, so its mutual information is $2I_{q,p}(\rho)$ and the mutual information per channel use is again $I_{q,p}(\rho)$. Equation~\eqref{eq:exact-constellation-mi} is a fixed-input quantity. It is therefore called constellation-constrained mutual information.

Since the input alphabet is finite, $I_{q,p}(\rho)$ approaches $\log_2 M$ as the SNR increases. The next result gives the speed of this convergence and an explicit bound on the error of the leading high-SNR approximation.  It is a specialisation of the arbitrary-input fading expansions in \cite{rodrigues2014multiple,ramos2014fading}, with an explicit remainder for the Gamma fading gain. Define the AWGN mutual information deficit $\eta_{q,p}(s) =\log_2M-I\!\left(X;\sqrt{s}X+Z\right).$ Here \(s\geq0\) is a scalar SNR. For each fixed \(s\), the expression \(\sqrt{s}X+Z\) defines an AWGN channel with input \(X\). Now define,

\begin{equation}
 \mu_j(q,p) =\int_0^{\infty}s^{2n_r+j-1}\eta_{q,p}(s)\,\mathrm ds.
 \label{eq:mi-deficit-moments}
\end{equation}
for $j\in\{0,1\}$.

\begin{Teo}\label{thm:controlled-mi-expansion}
For every fixed $q$ and $p$, the moments $\mu_0(q,p)$ and $\mu_1(q,p)$ are finite. Moreover,
$$0 \leq \frac{2^{2n_r}\mu_0(q,p)}{\Gamma(2n_r)\rho^{2n_r}} -\bigl[\log_2M-I_{q,p}(\rho)\bigr] \leq \frac{2^{2n_r+1}\mu_1(q,p)}{\Gamma(2n_r)\rho^{2n_r+1}}.$$
Consequently,
\begin{equation}
 \log_2M-I_{q,p}(\rho) =\frac{2^{2n_r}\mu_0(q,p)} {\Gamma(2n_r)}\rho^{-2n_r} +O(\rho^{-2n_r-1}).
 \label{eq:mi-leading-order}
\end{equation}
Thus, the mutual information deficit decreases as $\rho^{-2n_r}$.
\end{Teo}

\begin{proof}
The fading-averaged deficit is

\begin{equation*}
 \log_2M-I_{q,p}(\rho) =\mathbb E_h\!\left[\eta_{q,p}\!\left(\frac{\rho h}{2}\right)\right].
\end{equation*}

Since the density of $h$ is
$f_h(t)=t^{2n_r-1}e^{-t}/\Gamma(2n_r)$, the change of variables
$s=\rho t/2$ gives the exact identity

\begin{equation}
 \log_2M-I_{q,p}(\rho) =\frac{2^{2n_r}}{\Gamma(2n_r)\rho^{2n_r}} \int_0^{\infty}s^{2n_r-1}e^{-2s/\rho}\eta_{q,p}(s)\,\mathrm ds.
 \label{eq:mi-deficit-transform}
\end{equation}

It remains to justify the moments and the expansion.  Near zero, $0\leq\eta_{q,p}(s)\leq\log_2M$.  Let $d_{\min}>0$ be the minimum distance of $\mathcal A_{q,p}$. At large $s$, the scalar ML error
probability satisfies
\[
 P_{\mathrm e}(s) \leq(M-1)Q\!\left(\sqrt{\frac{s d_{\min}^2}{2}}\right) \leq\frac{M-1}{2}e^{-s d_{\min}^2/4}.
\]

Since $\eta_{q,p}(s)=H(X\mid\sqrt{s}X+Z)$, Fano's inequality shows that $\eta_{q,p}(s)$ is bounded by a polynomial times an exponentially decreasing function.  Hence all the moments in \eqref{eq:mi-deficit-moments} are finite.

For $x\geq0$, $0\leq1-e^{-x}\leq x$. Applying this inequality with $x=2s/\rho$ in \eqref{eq:mi-deficit-transform} gives the results.
\end{proof}

The leading coefficient has a direct operational meaning. Consider two energy-normalised alphabets $\mathcal A$ and $\mathcal B$ with the same cardinality. Let $\rho_{\mathcal A}(\delta)$ denote the SNR required to obtain the small deficit $\delta=\log_2M-I(\rho)$. Then
\begin{equation}
 \lim_{\delta\rightarrow0}10\log_{10}\frac{\rho_{\mathcal B}(\delta)}{\rho_{\mathcal A}(\delta)} = \frac{10}{2n_r}\log_{10}\frac{\mu_0(\mathcal B)}{\mu_0(\mathcal A)}.
 \label{eq:exact-asymptotic-mi-gap}
\end{equation}
This follows immediately by solving the leading term in \eqref{eq:mi-leading-order} for $\rho$.

Equation~\eqref{eq:exact-asymptotic-mi-gap} clarifies the role of the finite alphabet. The coefficient $\mu_0$ depends on its complete AWGN mutual information curve. The shaping and finite energy results above therefore do not by themselves imply larger mutual information at every SNR.

The fixed-alphabet coefficient in \eqref{eq:exact-asymptotic-mi-gap} replaces the distance-spectrum approximation used in \cite{souza2026AESTBC}.

\subsection{Error-probability}
\label{subsec:cer-implication}

We conclude this section by connecting the energy analysis with the CER results. For two distinct codewords, define $D=|\Delta x_0|^2+|\Delta x_1|^2,$ where the symbols belong to the energy normalised alphabet $\mathcal A_{q,p}$. Because of Alamouti orthogonality, the pairwise error probability depends on the codeword difference only through $D$.

Conditioned on $h$, coherent ML detection gives
\begin{equation*}
 P(\mathbf X\rightarrow\mathbf X'\mid h) = Q\!\left(\sqrt{\frac{\rho hD}{4}}\right).
 \label{eq:conditional-pep}
\end{equation*}
Using Craig's representation of the $Q$-function and averaging over $h\sim\Gamma(2n_r,1)$ gives the standard expression \cite{proakis1998fading}
\begin{align}
 P_{\mathrm{pep}}(D,\rho) &=\frac1\pi\int_0^{\pi/2} \left(1+\frac{\rho D}{8\sin^2\theta}\right)^{-2n_r} \mathrm d\theta \nonumber\\
 &=\kappa_{\nu}(\rho D)^{-\nu}+O(\rho^{-\nu-1}),
 \label{eq:exact-averaged-pep}
\end{align}
where $\nu=2n_r$ and $\kappa_{2n_r}=2^{2n_r-1}\binom{4n_r}{2n_r}$.  Therefore, every member of the proposed family achieves diversity order $2n_r$.

For a scalar alphabet with average energy $E$ and minimum distance $d_{\min}$, energy normalisation gives the minimum squared distance $d_{\min}^2/E$.  Since \eqref{eq:exact-averaged-pep} depends on $\rho D$, the SNRs required by two alphabets to obtain the same nearest-neighbour PEP satisfy
\begin{equation*}
 \frac{\rho_1}{\rho_2} =\frac{E_1/d_1^2}{E_2/d_2^2}.
 \label{eq:exact-nearest-neighbour-shift}
\end{equation*}
Theorem~\ref{thm:finite-family-optimality} therefore shows that the Eisenstein case has the smallest averaged nearest-neighbour PEP in the proposed family, for every $p$, under the common minimum distance convention.

The square and Eisenstein lattices both have $d_{\min}=1$. Thus, Proposition~\ref{prop:exact-finite-eisenstein-energy} gives the exact finite nearest-neighbour shift,
\begin{equation*}
 \frac{\rho_{\mathbb Z[\mathrm i]}}{\rho_{\Lambda_3}} =\frac{E_s(\mathbb Z[\mathrm i],p)}{E_s(\Lambda_3,p)} \rightarrow\frac65.
 \label{eq:eisenstein-nearest-neighbour-shift}
\end{equation*}

The limiting SNR advantage is therefore $10\log_{10}(6/5)\simeq0.792$ dB. This conclusion concerns a nearest-neighbour pair. It does not imply the same horizontal shift for the complete CER curve, which also depends on the remaining distances and their multiplicities. The complete CER comparison is therefore left to the numerical results section.

\section{Finite-Blocklength Analysis} \label{sec:finite-blocklength}

The Alamouti matrix is used as an inner space--time modulation block, and an outer code spans independently faded blocks. We retain the shared fading calculation of \cite{souza2026AESTBC}, derive the controlled information variance expansion, and use standard random-coding tools from \cite{polyanskiy2010finite} to describe its reliability implications.

The achievability statements below impose the finite alphabet input constraint. The energy convention \eqref{eq:stbc-energy-convention} holds for uniform modulation and for the i.i.d. random-coding ensemble. The extraction of a deterministic code from this ensemble does not impose a separate per-codeword energy bound or guarantee that its average energy equals the ensemble average.

\subsection{Independent Alamouti blocks}
\label{subsec:fbl-independent-blocks}

Let a packet span $m$ independently faded Alamouti blocks, or $n=2m$ complex channel uses. For a fixed alphabet $\mathcal A_{q,p}$ with $M=p^2$ points, consider $Y_s=\sqrt{s}X+Z$, where $X$ is uniform and $Z\sim\mathcal N_{\mathbb C}(0,1)$. Define its information density and variance by
\begin{equation*}
 \imath_s=\log_2\frac{p_{Y_s|X}(Y_s|X)}{p_{Y_s}(Y_s)},
 \qquad v(s)=\operatorname{Var}(\imath_s).
\end{equation*}
The alphabet indices are omitted in these two quantities. The equivalent channel in \eqref{eq:stbc-equivalent-scalar-channel} has $s=\rho h/2$. Using the AWGN deficit defined in Section~\ref{subsec:constellation-constrained-mi}, $\mathbb E[\imath_s]=\log_2M-\eta_{q,p}(s).$
At $s=\rho h/2$, averaging this conditional mean over $h$ gives $I_{q,p}(\rho)$.

Let $J_b$ be the sum of the two scalar information densities in block $b$. The variables $J_1,\ldots,J_m$ are independent and identically distributed.

\begin{Prop}\label{prop:fbl-block-dispersion}
For independent uniform inputs in the two Alamouti layers, $\mathbb E[J_b]=2I_{q,p}(\rho)$. The information variance per complex channel use is
\begin{align}
 V_{q,p}(\rho)
 &=\frac12\operatorname{Var}(J_b)=\mathbb E_h[v(\rho h/2)]
   +2\operatorname{Var}_h\bigl(\eta_{q,p}(\rho h/2)\bigr).
 \label{eq:fbl-effective-dispersion}
\end{align}
It is measured in bits squared per complex channel use. We also call it the fixed-input dispersion, without identifying it with the channel dispersion optimised over input distributions.
\end{Prop}

\begin{proof}
Conditioned on $h$, the two scalar information densities are independent and have the same distribution. Hence
\begin{align*}
 \mathbb E[J_b\mid h]&=2\bigl[\log_2M-\eta_{q,p}(\rho h/2)\bigr],\\
 \operatorname{Var}(J_b\mid h)&=2v(\rho h/2).
\end{align*}
The law of total expectation gives the mean. The law of total variance, divided by the two channel uses in a block, gives \eqref{eq:fbl-effective-dispersion}. Here the conditional mutual information and its deficit have the same variance, since their sum is the constant $\log_2M$.
\end{proof}

The factor two in the second term of \eqref{eq:fbl-effective-dispersion} accounts for the shared fading gain. The two scalar information densities are independent only after conditioning on $h$.

\subsection{High-SNR information variance}
The following result complements Theorem~\ref{thm:controlled-mi-expansion}. It applies to any fixed finite alphabet of distinct points with uniform independent inputs in the two Alamouti layers.

Put $\nu=2n_r$ and $a=2^\nu/\Gamma(\nu)$. Using the same $v(s)$ as above and the deficit $\eta_{q,p}(s)$, define, for $j=0,1$,
\begin{equation}
 \beta_j(q,p)=\int_0^\infty s^{\nu+j-1}
 \bigl[v(s)+2\eta_{q,p}(s)^2\bigr]\,\mathrm ds.
 \label{eq:variance-asymptotic-moments}
\end{equation}
The moments $\mu_j(q,p)$ remain as defined in \eqref{eq:mi-deficit-moments}.

\begin{Teo}\label{thm:controlled-variance-expansion}
For every fixed $q,p$ with $M=p^2\geq2$, both moments $\beta_0(q,p)$ and $\beta_1(q,p)$ are finite. For every $\rho>0$,
\begin{align}
 0&\leq a\beta_0(q,p)\rho^{-\nu}-V_{q,p}(\rho)
 \nonumber\\
 &\leq 2a\beta_1(q,p)\rho^{-\nu-1}
       +2a^2\mu_0(q,p)^2\rho^{-2\nu}.
 \label{eq:controlled-variance-remainder}
\end{align}
Consequently, as $\rho\to\infty$, $ V_{q,p}(\rho)\sim a\beta_0(q,p)\rho^{-2n_r},$
\begin{align*}
\frac{V_{q,p}(\rho)}{\log_2M-I_{q,p}(\rho)}
 &\longrightarrow\frac{\beta_0(q,p)}{\mu_0(q,p)}.
\end{align*}
\end{Teo}

\begin{proof}
Put $L_s=\log_2M-\imath_s$. Conditioned on $X=x$,
\begin{equation*}
 L_s=\log_2\!\left(1+\sum_{x'\ne x}
 e^{-|Z+\sqrt{s}(x-x')|^2+|Z|^2}\right).
\end{equation*}
For $t\geq0$, integrating $(1+u)^{-1}\leq u^{-3/4}$ gives $\ln(1+t)\leq4t^{1/4}$, and hence $\ln^2(1+t)\leq16\sqrt t$. Together with $\sqrt{\sum_k t_k}\leq\sum_k\sqrt{t_k}$ for $t_k\geq0$, Gaussian integration yields
\begin{equation*}
 \mathbb E[L_s^2]\leq
 \frac{16(M-1)}{(\ln2)^2}e^{-s d_{\min}^2/4},
\end{equation*}
where $d_{\min}>0$ is the alphabet's minimum distance. Since $\eta_{q,p}(s)=\mathbb E[L_s]$ and $v(s)=\operatorname{Var}(L_s)$, $ v(s)+2\eta_{q,p}(s)^2\leq2\mathbb E[L_s^2].$ Thus both moments in \eqref{eq:variance-asymptotic-moments} are finite.

Expanding the variance in \eqref{eq:fbl-effective-dispersion} and substituting $s=\rho h/2$ gives
\begin{align*}
 V_{q,p}(\rho)
 &=a\rho^{-\nu}\int_0^\infty s^{\nu-1}e^{-2s/\rho}\bigl[v(s)+2\eta_{q,p}(s)^2\bigr]\,\mathrm ds\\
 &\qquad-2\bigl[\log_2M-I_{q,p}(\rho)\bigr]^2.
\end{align*}
Subtracting this identity from $a\beta_0(q,p)\rho^{-\nu}$, and using $0\leq1-e^{-x}\leq x$ and
\begin{equation*}
 0\leq\log_2M-I_{q,p}(\rho)\leq a\mu_0(q,p)\rho^{-\nu},
\end{equation*}
proves \eqref{eq:controlled-variance-remainder}. The asymptotic statements follow from this bound and \eqref{eq:mi-leading-order}, since $\mu_0(q,p),\beta_0(q,p)>0$ for $M\geq2$.
\end{proof}

Thus the information variance and mutual information deficit have the same diversity exponent, with an explicit limiting ratio determined by different AWGN moments. In the normal approximation below, the deficit gives the gap of the first-order rate to $\log_2M$, while the information variance gives the finite-length penalty. These quantities use the same Alamouti combining and scalar channel model.

\subsection{Achievable rates}
We apply standard random-coding and second-order tools to the information quantities above. Let $L$ be the number of messages and $S_m=\sum_{b=1}^m J_b$ the information density of one packet.

\begin{Prop}\label{prop:fbl-rcus}
For every $m\geq1$ and integer $L\geq2$, there exists an alphabet constrained code with $L$ messages and average packet error probability satisfying
\begin{equation}
 \epsilon\leq\mathbb E\!\left[\min\{1,(L-1)2^{-S_m}\}\right].
 \label{eq:fbl-rcus}
\end{equation}
The expectation is over the uniform inputs, fading, and noise.
\end{Prop}

\begin{proof}
Treat one Alamouti block as a memoryless channel use whose input is the pair of transmitted symbols and whose output includes $h_b$. The random-coding union bound in \cite{polyanskiy2010finite}, followed by Markov's inequality for the competing codeword, gives \eqref{eq:fbl-rcus}.
\end{proof}

This bound can be evaluated by Monte Carlo. A simpler expression follows from the second-order approximation.

\begin{Teo}\label{thm:fbl-second-order}
Fix $q$, $p$, $\rho>0$, and $\epsilon\in(0,1)$, with $V_{q,p}(\rho)>0$. There are finite constants $K$ and $n_0$ such that, for every even $n\geq n_0$, an alphabet constrained outer code exists with average error probability at most $\epsilon$ and rate $R=(\log_2L)/n$ satisfying
\begin{equation}
 R\geq I_{q,p}(\rho)-\sqrt{\frac{V_{q,p}(\rho)}{n}}Q^{-1}(\epsilon)
       -K\frac{\log_2n}{n}.
 \label{eq:fbl-second-order-achievability}
\end{equation}
The rate is measured in bits per complex channel use.
\end{Teo}

\begin{proof}
The sum $S_m$ has mean $nI_{q,p}(\rho)$ and variance $nV_{q,p}(\rho)$. Each $J_b$ has a finite third absolute moment, since the alphabet is finite and the Gaussian noise and Gamma fading have finite moments of every order. For any $\tau>0$, \eqref{eq:fbl-rcus} gives
\begin{equation*}
 \epsilon\leq\Pr\!\left\{S_m\leq\log_2(L-1)+\tau\right\}+2^{-\tau}.
\end{equation*}
Taking $\tau=\log_2m$ and applying the Berry--Esseen theorem to $S_m$, a sufficiently large logarithmic back-off gives \eqref{eq:fbl-second-order-achievability}. This is the second-order random-coding argument of \cite{polyanskiy2010finite} applied to independent Alamouti blocks.
\end{proof}

Omitting the logarithmic term gives the normal approximations
\begin{align*}
 R_{\mathrm{NA}}^{(q,p)}
 &=I_{q,p}(\rho)-\sqrt{\frac{V_{q,p}(\rho)}{n}}Q^{-1}(\epsilon),\\
 \epsilon_{\mathrm{NA}}^{(q,p)}
 &=Q\!\left(\sqrt{\frac{n}{V_{q,p}(\rho)}}
       \bigl[I_{q,p}(\rho)-R\bigr]\right).
\end{align*}
The dependence on $n$, $\rho$, and the target error probability or rate is left implicit. These are approximations, equation \eqref{eq:fbl-rcus} gives a nonasymptotic achievability bound.

At common $n$, $\epsilon$, and $\rho$, comparing the Eisenstein case with another $q$ gives
\begin{align*}
 R_{\mathrm{NA}}^{(3,p)}-R_{\mathrm{NA}}^{(q,p)}
 &=I_{3,p}(\rho)-I_{q,p}(\rho)\\
 &\quad-\frac{Q^{-1}(\epsilon)}{\sqrt n}
 \left(\sqrt{V_{3,p}(\rho)}-\sqrt{V_{q,p}(\rho)}\right).
\end{align*}
The first term is the mutual information contribution and the second is the dispersion contribution. For $\epsilon<1/2$, if $I_{3,p}(\rho)\geq I_{q,p}(\rho)$ and $V_{3,p}(\rho)\leq V_{q,p}(\rho)$, the normal approximation favours the Eisenstein case for every $n$.

\section{Numerical Results} \label{sec:numerical-results}

This section evaluates the finite-alphabet energy, constellation-constrained mutual information, codeword error rate, and finite-blocklength performance of the proposed family.  We first describe the common numerical setup and verify the finite-energy results of Section~\ref{subsec:shaping-energy}.

\subsection{Simulation setup and finite-energy comparison} \label{subsec:numerical-setup-energy}

We consider the Eisenstein case $q=3$, the cases $q=7$ and $q=11$, and the square lattice $\mathbb Z[\mathrm i]$.  The case $q=3$ is the proposed optimum, while $q=7$ is the next case of the quaternionic family. The case $q=11$ is included to confirm the result in Theorem~\ref{thm:finite-family-optimality}.  The communication results use $p=4$ and $M=16$. The energy comparison also includes $p=8$ and $p=16$ to show the approach to the asymptotic result.

For all channel simulations, we use the model in \eqref{eq:stbc-channel}. The entries of $\mathbf H_b$ and $\mathbf Z_b$ are independent $\mathcal{CN}(0,1)$ random variables.  The channel remains constant during the two time slots of one Alamouti block and changes independently between blocks.  Perfect channel knowledge is available at the receiver, while the transmitter has no channel knowledge.  Unless stated otherwise, $n_r=1$, and the horizontal axis is $10\log_{10}\rho$, where $\rho$ is the average receive SNR per receive antenna.  The same channel and noise realisations are used across the compared alphabets whenever possible.

To verify the analytical energy results, we use the quotient alphabet defined in Section~\ref{sec:preliminaries},
\begin{equation}
 \mathcal X_{\Lambda,p}=\Lambda\cap p\mathcal V_0(\Lambda), \qquad \mathcal A_{\Lambda,p} =\frac{\mathcal X_{\Lambda,p}}{\sqrt{E_s(\Lambda,p)}}.
 \label{eq:numerical-exact-alphabet}
\end{equation}

One minimum norm representative is selected from every coset of $\Lambda/p\Lambda$. Ties on the Voronoi boundary are resolved by choosing the lexicographically smallest pair of lattice coefficients $(m,n)$.  This fixed rule gives $|\mathcal X_{\Lambda,p}|=p^2$. The energy is computed by exact rational arithmetic and is then used in \eqref{eq:numerical-exact-alphabet}. Consequently, $\mathbb E[|X|^2]=1$ and $\mathbb E[\|\mathbf X\|_{\mathrm F}^2]=4$ for every code.

The exact quotient representatives need not have zero mean. No centring is applied to the results in Table~\ref{tab:finite-energy-comparison}, because the purpose of this table is to verify Proposition~\ref{prop:exact-finite-eisenstein-energy} and Theorem~\ref{thm:finite-family-optimality}. In the later comparison with conventional square QAM and full-rate space--time codes, every modulation alphabet is instead translated to zero mean and renormalised. 

For a lattice $\Lambda$, define the finite energy advantage over the square lattice by
\begin{equation*}
 \gamma_{E,p}(\Lambda) =10\log_{10}\frac{E_s(\mathbb Z[\mathrm i],p)}{E_s(\Lambda,p)}\quad\text{dB}.
 \label{eq:numerical-finite-energy-gain}
\end{equation*}
Thus, a positive value denotes an energy advantage. The values in Table~\ref{tab:finite-energy-comparison} are obtained by enumerating the $p^2$ cosets. For the square and Eisenstein lattices, the enumeration also
agrees exactly with Proposition~\ref{prop:exact-finite-eisenstein-energy}.

\begin{table*}[htpb]
\caption{Exact finite-alphabet energy at fixed cardinality and minimum distance.} 
\label{tab:finite-energy-comparison}
\centering
\footnotesize
\setlength{\tabcolsep}{5pt}
\begin{tabular}{cccccccc}
\hline
Lattice & $E_s(\Lambda,4)$ & $\gamma_{E,4}$ (dB) & $E_s(\Lambda,8)$ & $\gamma_{E,8}$ (dB) & $E_s(\Lambda,16)$ & $\gamma_{E,16}$ (dB) & $\gamma_E$ (dB)\\
\hline
$\mathbb Z[\mathrm i]$ & $3$ & $0$ & $11$ & $0$ & $43$ & $0$ & $0$\\
$\Lambda_3$ & $9/4$ & $1.249$ & $9$ & $0.872$ & $285/8$ & $0.817$ & $0.792$\\
$\Lambda_7$ & $15/4$ & $-0.969$ & $29/2$ & $-1.200$ & $58$ & $-1.300$ & $-1.326$\\
$\Lambda_{11}$ & $21/4$ & $-2.430$ & $20$ & $-2.596$ & $159/2$ & $-2.669$ & $-2.704$\\
\hline
\end{tabular}
\end{table*}

The results confirm the finite-alphabet ordering proved in Theorem~\ref{thm:finite-family-optimality}.  For all tested values of $p$, $\Lambda_3$ has the smallest energy among the proposed cases, whereas $\Lambda_7$ and $\Lambda_{11}$ require more energy than the square lattice. For $p=4$, the exact alphabet advantage of $\Lambda_3$ is $1.249$ dB. It decreases to $0.872$ dB for $p=8$ and $0.817$ dB for $p=16$, approaching the asymptotic value $10\log_{10}(6/5)\simeq0.792$ dB. These finite-energy gains are different from the conventional equal volume shaping gain of $0.167$ dB. 

\subsection{Mutual information and high-SNR behaviour}
\label{subsec:numerical-mi}

We evaluate \eqref{eq:exact-constellation-mi} for the uncentred quotient alphabets with $p=4$, $M=16$, and $n_r=1$, comparing $\Lambda_3$, the square alphabet, and $\Lambda_7$. The average over input symbols is exact. Gaussian noise is integrated by Gauss--Hermite quadrature with $32$ nodes per real dimension, while \eqref{eq:mi-deficit-transform} uses $360$ Gauss--Legendre nodes after the substitution $s=t/(1-t)$. Increasing these orders to $36$ and $440$ changes the mutual information by less than $2\times10^{-5}$ bits per channel use and $\mu_0$ by less than $10^{-4}$ in relative value.

The same quadrature gives the moments in \eqref{eq:mi-deficit-moments}. The fitted exponents use the last four SNR points between $34$ and $40$ dB:

\begin{equation*}
\begin{array}{c|ccc}
 \text{Alphabet} & \mu_0 & \mu_1 & \text{exponent}\\ \hline
 \Lambda_3             &195.904&3251.642&-1.993\\
 \mathbb Z[\mathrm i]  &295.351&6233.034&-1.991\\
 \Lambda_7             &285.665&6624.778&-1.990
\end{array}
\end{equation*}

The exponents are close to $-2$, as predicted by Theorem~\ref{thm:controlled-mi-expansion}. For $n_r=1$, its bound is

\begin{equation}
 0\leq 4\mu_0(q,p)-\rho^2\!\left[\log_2M-I_{q,p}(\rho)\right] \leq\frac{8\mu_1(q,p)}{\rho}.
 \label{eq:numerical-controlled-mi}
\end{equation}

Figure~\ref{fig:mi-high-snr} shows the scaled deficit approaching $4\mu_0(q,p)$ from below, consistently with \eqref{eq:numerical-controlled-mi}. The limiting Eisenstein advantage over the square alphabet is
 \[
 \gamma_{\mathrm{MI},4} =5\log_{10}\frac{\mu_0(\mathbb Z[\mathrm i],4)}{\mu_0(3,4)}\simeq0.89148\ \mathrm{dB}.
 \]
 
This gain applies to the fixed, uncentred alphabets. Although $\Lambda_7$ has greater finite energy than the square alphabet, its smaller $\mu_0$ gives a smaller leading high-SNR deficit. Thus energy ordering alone does not determine mutual information ordering.

\begin{figure}[t]
 \centering
 \includegraphics[width=0.98\columnwidth]{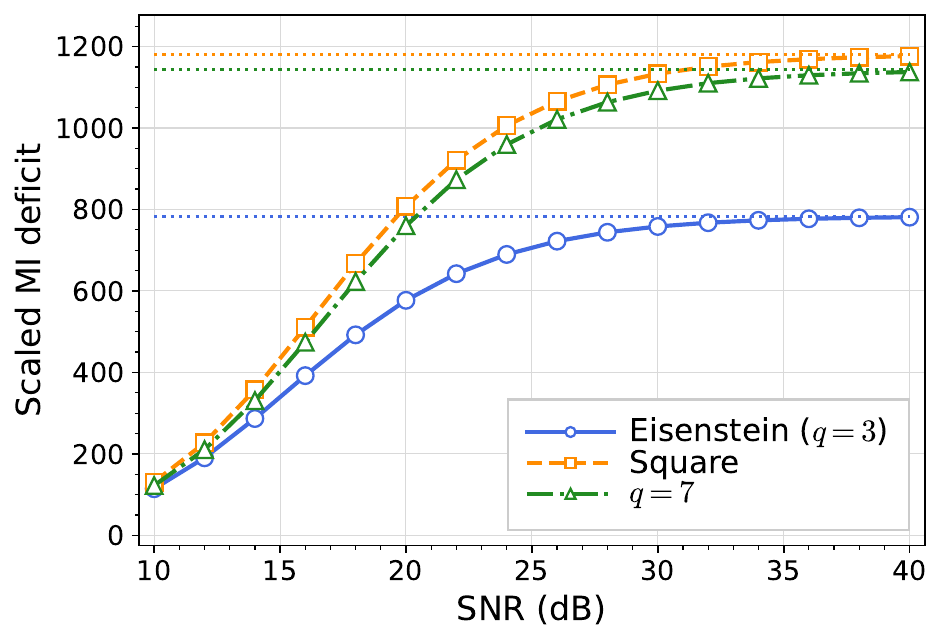}
 \caption{Scaled mutual information deficit $\rho^2[\log_2M-I_{q,p}(\rho)]$ for the exact, uncentred $p=4$ quotient alphabets and $n_r=1$. The coloured dotted levels give the leading coefficients $4\mu_0(q,p)$.}
 \label{fig:mi-high-snr}
\end{figure}

\subsection{Codeword-error-rate performance}
\label{subsec:numerical-cer}

We compare the same three alphabets with $p=4$ and $n_r=1$, at $4$ bits per channel use and unit average symbol energy. The equivalent channel in \eqref{eq:stbc-equivalent-scalar-channel} uses one fading gain for both symbols and independent noise samples. Detection follows \eqref{eq:stbc-symbolwise-decoding}, and a codeword error occurs if either symbol is detected incorrectly. We consider the uncentred quotient alphabets and their centred, renormalised versions,

\begin{equation}
 \mathcal A^{\mathrm c}_{\Lambda,p}
 =\frac{\mathcal X_{\Lambda,p}-\bar x_{\Lambda,p}}
 {\sqrt{E_s(\Lambda,p)-|\bar x_{\Lambda,p}|^2}},\qquad
 \bar x_{\Lambda,p}=\frac1M\sum_{x\in\mathcal X_{\Lambda,p}}x.
 \label{eq:numerical-centred-alphabet}
\end{equation}

The simulations use $10\log_{10}\rho\in\{10,12,\ldots,32\}$ dB and fixed sample sizes selected in an independent pilot calculation, from $10^5$ to $5\times10^6$ blocks. Common symbol indices, fading gains, and noise samples are used across alphabets. 

\begin{table}[htpb]
\caption{Required SNR in dB at a codeword-error rate of $10^{-3}$ for $p=4$ and $n_r=1$, obtained by logarithmic interpolation.}
\label{tab:numerical-cer-summary}
\centering
\footnotesize
\setlength{\tabcolsep}{7pt}
\begin{tabular}{lcc}
\hline
Alphabet & Uncentred & Centred\\
\hline
$\Lambda_3$ & $27.65$ & $27.52$\\
$\mathbb Z[\mathrm i]$ & $28.59$ & $27.79$\\
$\Lambda_7$ & $28.55$ & $28.47$\\
\hline
\end{tabular}
\end{table}

Table~\ref{tab:numerical-cer-summary} gives the SNRs at a CER of $10^{-3}$, obtained by logarithmic interpolation. The Eisenstein case requires the lowest SNR under both conventions. Centring reduces its advantage over the square alphabet from approximately $1$ dB to $0.3$ dB; its advantage over $\Lambda_7$ remains approximately $0.9$ dB. Thus the finite-alphabet mean affects the complete-code comparison.

\subsection{Finite-blocklength performance}
\label{subsec:numerical-fbl}

We use the exact, uncentred $p=4$ quotient alphabets, with $M=16$ and $n_r=1$. Each packet spans $m$ independently faded Alamouti blocks, giving $n=2m$ complex channel uses. The two symbols in a block share the gain $h_b$.

Information-density moments are averaged over $h_b\sim\Gamma(2,1)$ using $24$ Gauss--Hermite nodes per real noise dimension and piecewise Gauss--Legendre quadrature with $14$ nodes per fading interval.
We calculate $I_{q,4}(\rho)=\frac12\mathbb E[J_b],$ and $ V_{q,4}(\rho)=\frac12\operatorname{Var}(J_b).$
At $10$ dB, treating the layers as independently faded would omit about $21\%$ of the information variance.

\begin{table}[htpb]
\caption{Finite-blocklength quantities for the exact $p=4$ alphabets at $10$ dB, $n=128$, and $\epsilon=10^{-2}$. Rates are in bits per complex channel use.}
\label{tab:numerical-fbl-summary}
\centering
\footnotesize
\setlength{\tabcolsep}{3.5pt}
\begin{tabular}{ccccc}
\hline
Alphabet & $I$ & $V$ & $R_{\mathrm{NA}}$ & $R_{\rm (RCU)}$\\
\hline
$\Lambda_3$            & $2.855101$ & $2.608264$ & $2.523019$ & $2.484375$\\
$\mathbb Z[\mathrm i]$ & $2.691079$ & $2.666319$ & $2.355321$ & $2.320313$\\
$\Lambda_7$            & $2.775995$ & $2.571285$ & $2.446276$ & $2.414063$\\
\hline
\end{tabular}
\end{table}

Figure~\ref{fig:fbl-family-p4} shows that the variance ordering changes with SNR: $\Lambda_3$ has the largest value at low SNR and the smallest at higher SNR, whereas $\Lambda_7$ has the smallest value at $10$ dB. Energy optimality therefore does not imply minimum variance at every SNR.

In Table~\ref{tab:numerical-fbl-summary}, $R_{\mathrm{NA}}$ is the second order normal-approximation rate. For the relaxed RCU bound in \eqref{eq:fbl-rcus}, we generated $3\times10^5$ independent packets at $n=128$ and tested $L=2^k$ messages. We define $R_{\rm (RCU)}$ as the largest tested rate whose empirical-Bernstein upper confidence limit is at most $10^{-2}$, with simultaneous $99\%$ confidence over the tested rates. Both evaluations favour $\Lambda_3$. For the normal approximation, its larger mutual information offsets its higher variance relative to $\Lambda_7$ at this SNR. The RCU values evaluate a nonasymptotic random-coding bound, rather than an implemented outer code.

\begin{figure}[t]
 \centering
 \includegraphics[width=0.98\columnwidth]{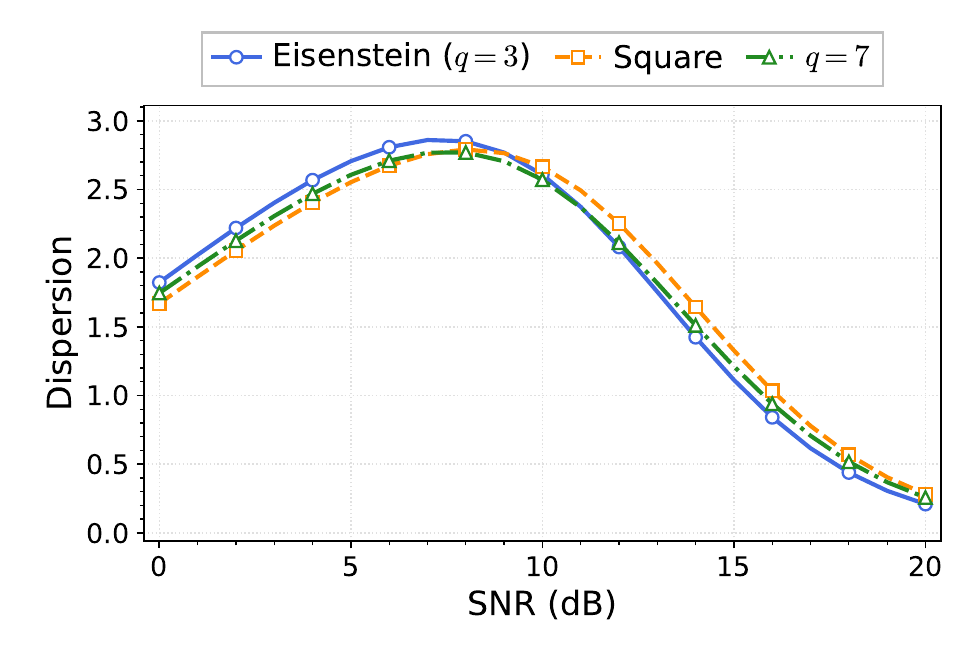}
 \caption{Constellation-constrained information variance versus SNR for the exact, uncentred $p=4$ quotient alphabets and $n_r=1$.}
 \label{fig:fbl-family-p4}
\end{figure}

\subsection{Comparison with standard space--time codes}
\label{subsec:numerical-standard-stbc}

We compare classical Alamouti, $\Lambda_3$, $\Lambda_7$, Golden~\cite{belfiore2005golden}, and Silver~\cite{hollanti2008silver} over the coherent block-fading channel in \eqref{eq:stbc-channel}, with $n_r=2$. The Alamouti-type codes use centred 16-point alphabets from \eqref{eq:numerical-centred-alphabet}; Golden and Silver use unit-energy QPSK. The alphabet sizes are $M_{\rm A}=16$ and $M_{\rm F}=4$, respectively, with $\log_2M_{\rm A}=2\log_2M_{\rm F}$: every code contains $256$ matrices and transmits $4$ bits per channel use. We normalise each codebook by $\widetilde{\mathbf X}=\sqrt{\frac{4}{\mathbb E[\|\mathbf X\|_{\rm F}^{2}]}}\mathbf X,$ so that $\mathbb E[\|\mathbf X\|_{\rm F}^2]=4$, using the canonical Golden factor $1/\sqrt5$ and the factor $1/\sqrt2$ for Silver's sum of two Alamouti matrices.

Our exact ML implementation uses two 16-point searches for each Alamouti-type code and direct searches over $256$ matrices for Golden and Silver. The separated Alamouti decisions agree with direct ML. All codebooks contain $256$ distinct matrices and satisfy the average energy convention in \eqref{eq:stbc-energy-convention}.

Figure~\ref{fig:numerical-standard-stbc} uses $10^5$ codewords per point from $6$ to $14$ dB, $3\times10^5$ at $16$ dB, and $10^6$ at $18$ dB. 

Table~\ref{tab:numerical-standard-stbc} gives logarithmically interpolated SNRs at $P_{\rm C}=10^{-2}$. Eisenstein improves classical Alamouti by $0.278$ dB and $\Lambda_7$ by $0.983$ dB; Golden and Silver improve Eisenstein by $1.912$ and $1.706$ dB. All five codes have diversity order $4$.

\begin{table}[t]
\caption{Equal-rate comparison at $4$ bits per channel use.  The last column gives the number of distances evaluated by the straightforward exact decoder used in the simulation.}
\label{tab:numerical-standard-stbc}
\centering
\footnotesize
\setlength{\tabcolsep}{3pt}
\begin{tabular}{lccc}
\hline
Code & Input alphabet & SNR at $10^{-2}$ & Direct metrics\\
\hline
$\Lambda_3$ Alamouti & 16-point & $15.937$ dB & $32$ scalar\\
Classical Alamouti & 16-QAM & $16.216$ dB & $32$ scalar\\
$\Lambda_7$ Alamouti & 16-point & $16.921$ dB & $32$ scalar\\
Golden & QPSK & $14.025$ dB & $256$ matrix\\
Silver & QPSK & $14.231$ dB & $256$ matrix\\
\hline
\end{tabular}
\end{table}

The metric counts describe these implementations, not inherent complexity. Square QAM permits rectangular slicing, whereas $\Lambda_3$ and $\Lambda_7$ use two-dimensional quantisation. Fast exact Golden and Silver decoders have worst-case orders $O(M_{\rm F}^{2.5})$ and $O(M_{\rm F}^{2})$~\cite{sinnokrot2010golden,srinath2011silver}. The approximately $0.28$ dB improvement over classical Alamouti supports the finite-alphabet comparison while retaining two separate symbol decisions. Golden and Silver achieve lower CER with coupled detection.

\begin{figure}[t]
 \centering
 \includegraphics[width=0.98\columnwidth]
 {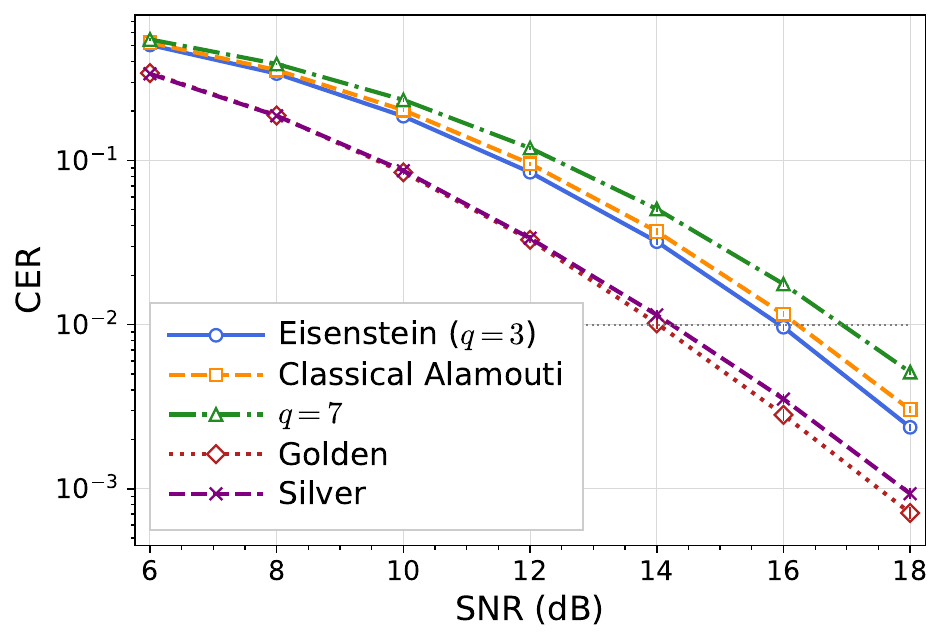}
 \caption{Exact-ML codeword-error rate over the coherent $2\times2$ Rayleigh block-fading channel at $4$ bits per channel use. The Alamouti-type codes use centred 16-point alphabets, while the Golden and Silver codes use QPSK.  All codebooks have equal cardinality and equal average transmitted energy.}
 \label{fig:numerical-standard-stbc}
\end{figure}

\section{Conclusion}\label{sec:conclusion}

This work studied finite Alamouti-type codes associated with a classical family of maximal quaternion orders. The main shaping results give exact finite energy optimality of the Eisenstein case for every $p\geq2$, at fixed alphabet size and minimum lattice distance. The ordering remains valid when each member jointly optimises its quotient representatives and translation, with an explicit positive gap for every $q>3$. These comparisons retain the common orthogonal encoder and symbol-wise ML detector.

For uniform finite inputs, the controlled information variance expansion gives its high-SNR coefficient, remainder, and limiting ratio to the mutual information deficit, including the shared fading contribution. The mutual information expansion specialises arbitrary input fading results. Standard random-coding and second-order tools give the corresponding reliability interpretation over independently faded Alamouti blocks. The results do not establish mutual information or finite-blocklength optimality of the Eisenstein case at every SNR.

Numerical results support the finite-alphabet analysis. In the centred equal rate comparison, Eisenstein improves classical Alamouti by approximately $0.28$ dB at a CER of $10^{-2}$. Golden and Silver require about $1.9$ dB and $1.7$ dB less SNR, respectively, with coupled detection. This comparison places the shaping benefit in the context of reliability and decoding effort; the measured gain supports the analysis while two separate symbol decisions are retained.

\section*{AI Use Statement}

The authors used OpenAI ChatGPT to assist with editing, specifically Fig. 1 and Fig. 2, and with grammar enhancement throughout the paper. The author reviewed and edited all output and takes full responsibility for the content of this publication.

\renewcommand\refname{References}
\bibliographystyle{ieeetr}
\bibliography{references}

@BOOK{Con2013,
	title={Sphere Packings, Lattices and Groups},
	author={Conway, J. H. and Sloane, N. J. A.},
	edition={3},
	series={Grundlehren der mathematischen Wissenschaften},
	volume={290},
	year={2013},
	publisher={Springer},
	address={New York, NY},
    doi={10.1007/978-1-4757-6568-7}}

@article{huang2017lattice,
	title={Lattice index codes from algebraic number fields},
	author={Huang, Yu-Chih},
	journal={IEEE Transactions on Information Theory},
	volume={63},
	number={4},
	pages={2098--2112},
	year={2017},
	doi={10.1109/TIT.2017.2655034},
	publisher={IEEE}
}

@article{wang2012quadratic,
  author={Wang, Genyuan and Zhang, Jian-Kang and Amin, Moeness G.},
  title={Space-Time Block Code Designs Based on Quadratic Field Extension for Two-Transmitter Antennas},
  journal={IEEE Transactions on Information Theory},
  year={2012},
  month=jun,
  volume={58},
  number={6},
  pages={4005--4013},
  doi={10.1109/TIT.2012.2184633}
}

@article{huang2018layered,
	 author={Huang, Yu-Chih and Hong, Yi and Viterbo, Emanuele and Natarajan, Lakshmi},
  journal={IEEE Transactions on Information Theory}, 
  title={Layered Space-Time Index Coding}, 
  year={2019},
  volume={65},
  number={1},
  pages={142-158},
  doi={10.1109/TIT.2018.2842144}}

@book{voight2021quaternion,
	title={Quaternion Algebras},
	author={Voight, John},
	series={Graduate Texts in Mathematics},
	volume={288},
	year={2021},
	publisher={Springer},
	address={Cham},
doi={10.1007/978-3-030-56694-4}
}

@article{huang2018lattices,
	  author={Huang, Yu-Chih and Narayanan, Krishna R. and Wang, Ping-Chung},
  journal={IEEE Transactions on Information Theory}, 
  title={Lattices Over Algebraic Integers With an Application to Compute-and-Forward}, 
  year={2018},
  volume={64},
  number={10},
  pages={6863-6877},
  doi={10.1109/TIT.2018.2848277}}

@ARTICLE{hollanti2008maximal,
  author={Hollanti, Camilla and Lahtonen, Jyrki and Lu, Hsiao-feng},
  journal={IEEE Transactions on Information Theory}, 
  title={Maximal Orders in the Design of Dense Space-Time Lattice Codes}, 
  year={2008},
  volume={54},
  number={10},
  pages={4493-4510},
  doi={10.1109/TIT.2008.928998}}

@ARTICLE{vehkalahti2009densest,
  author={Vehkalahti, Roope and Hollanti, Camilla and Lahtonen, Jyrki and Ranto, Kalle},
  journal={IEEE Transactions on Information Theory}, 
  title={On the Densest {MIMO} Lattices From Cyclic Division Algebras}, 
  year={2009},
  volume={55},
  number={8},
  pages={3751-3780},
  doi={10.1109/TIT.2009.2023713}}

@article{oggier2006perfect,
   author={Oggier, F. and Rekaya, G. and Belfiore, J.-C. and Viterbo, E.},
  journal={IEEE Transactions on Information Theory}, 
  title={Perfect Space–Time Block Codes}, 
  year={2006},
  volume={52},
  number={9},
  pages={3885-3902},
  doi={10.1109/TIT.2006.880010}}

@ARTICLE{elia2006explicit,
  author={Elia, P. and Kumar, K.R. and Pawar, S.A. and Kumar, P.V. and Lu, H.-F.},
  journal={IEEE Transactions on Information Theory}, 
  title={Explicit Space–Time Codes Achieving the Diversity–Multiplexing Gain Tradeoff}, 
  year={2006},
  volume={52},
  number={9},
  pages={3869-3884},
  doi={10.1109/TIT.2006.880037}
}

@ARTICLE{alamouti1998,
  author={Alamouti, S.M.},
  journal={IEEE Journal on Selected Areas in Communications}, 
  title={A simple transmit diversity technique for wireless communications}, 
  year={1998},
  volume={16},
  number={8},
  pages={1451-1458},
  doi={10.1109/49.730453}}

@article{alves2015lattices,
 title = {Lattices from maximal orders into quaternion algebras},
journal = {Journal of Pure and Applied Algebra},
volume = {219},
number = {4},
pages = {687-702},
year = {2015},
issn = {0022-4049},
doi = {10.1016/j.jpaa.2014.04.025},
url = {https://www.sciencedirect.com/science/article/pii/S0022404914000942},
author = {C. Alves and J.-C. Belfiore}
}

@ARTICLE{proakis1998fading,
  author={Biglieri, E. and Proakis, J. and Shamai, S.},
  journal={IEEE Transactions on Information Theory}, 
  title={Fading channels: information-theoretic and communications aspects}, 
  year={1998},
  volume={44},
  number={6},
  pages={2619-2692},
  doi={10.1109/18.720551}}

@book{berhuy2013central,
author = {Berhuy, Grgory and Oggier, Frdrique},
title = {An Introduction to Central Simple Algebras and Their Applications to Wireless Communication},
year = {2013},
isbn = {0821849379},
publisher = {American Mathematical Society},
address = {USA}
}

@BOOK{oggier2007cyclic,
  author={Oggier, Frédérique and Belfiore, Jean-Claude and Viterbo, Emanuele},
  title={Cyclic Division Algebras: A Tool for Space-Time Coding},
  year={2007},
  volume={},
  number={},
  pages={},
  publisher = {Now Foundations and Trends},
  doi={10.1561/0100000016}
}

@ARTICLE{tarokh1999stbc,
  author={Tarokh, V. and Jafarkhani, H. and Calderbank, A.R.},
  journal={IEEE Transactions on Information Theory}, 
  title={Space-time block codes from orthogonal designs}, 
  year={1999},
  volume={45},
  number={5},
  pages={1456-1467},
  doi={10.1109/18.771146}}

@ARTICLE{tarokh1998criteria,
  author={Tarokh, V. and Seshadri, N. and Calderbank, A.R.},
  journal={IEEE Transactions on Information Theory}, 
  title={Space-time codes for high data rate wireless communication: performance criterion and code construction}, 
  year={1998},
  volume={44},
  number={2},
  pages={744-765},
  doi={10.1109/18.661517}}

@ARTICLE{rodrigues2014multiple,
  author={Rodrigues, Miguel R. D.},
  journal={IEEE Transactions on Information Theory}, 
  title={Multiple-Antenna Fading Channels With Arbitrary Inputs: Characterization and Optimization of the Information Rate}, 
  year={2014},
  volume={60},
  number={1},
  pages={569-585},
  doi={10.1109/TIT.2013.2287507}}

@ARTICLE{polyanskiy2010finite,
  author={Polyanskiy, Y. and Poor, H. V. and Verdu, S.},
  journal={IEEE Transactions on Information Theory}, 
  title={Channel Coding Rate in the Finite Blocklength Regime}, 
  year={2010},
  volume={56},
  number={5},
  pages={2307-2359},
  doi={10.1109/TIT.2010.2043769}}

@book{paulraj2003,
  author    = {Arogyaswami Paulraj and Rohit Nabar and Dhananjay Gore},
  title     = {Introduction to Space--Time Wireless Communications},
  publisher = {Cambridge University Press},
  address   = {Cambridge, UK},
  year      = {2003},
  isbn      = {978-0-521-82615-1}
}

@ARTICLE{durisi2016short-packet,
  author={Durisi, Giuseppe and Koch, Tobias and Östman, Johan and Polyanskiy, Yury and Yang, Wei},
  journal={IEEE Transactions on Communications}, 
  title={Short-Packet Communications Over Multiple-Antenna Rayleigh-Fading Channels}, 
  year={2016},
  volume={64},
  number={2},
  pages={618-629},
  doi={10.1109/TCOMM.2015.2511087}
}

@INPROCEEDINGS{souza2026AESTBC,
  author={Souza, Juliana G. F. and Huang, Yu-Chih},
  booktitle={2026 IEEE International Symposium on Information Theory (ISIT)}, 
  title={Finite-Blocklength Analysis of Alamouti Codes over Eisenstein Integers}, 
  year={2026},
  volume={},
  number={},
  pages={1-6},
  doi={10.1109/ISIT62367.2026.11653867}}

@ARTICLE{collins2019coherent,
  author={Collins, Austin and Polyanskiy, Yury},
  journal={IEEE Transactions on Information Theory}, 
  title={Coherent Multiple-Antenna Block-Fading Channels at Finite Blocklength}, 
  year={2019},
  volume={65},
  number={1},
  pages={380-405},
  doi={10.1109/TIT.2018.2860979}}

@ARTICLE{belfiore2005golden,
  author={Belfiore, J.-C. and Rekaya, G. and Viterbo, E.},
  journal={IEEE Transactions on Information Theory}, 
  title={The {Golden} Code: A {$2\times2$} Full-Rate Space-Time Code With Nonvanishing Determinants},
  year={2005},
  volume={51},
  number={4},
  pages={1432-1436},
  doi={10.1109/TIT.2005.844069}}

@INPROCEEDINGS{hollanti2008silver,
  author={Hollanti, C. and Lahtonen, J. and Ranto, K. and Vehkalahti, R. and Viterbo, E.},
  booktitle={2008 IEEE Information Theory Workshop}, 
  title={On the Algebraic Structure of the {Silver} Code: A {$2\times2$} Perfect Space-Time Block Code},
  year={2008},
  volume={},
  number={},
  pages={91-94},
  doi={10.1109/ITW.2008.4578628}}

@ARTICLE{sinnokrot2010golden,
  author={Sinnokrot, Mohanned O. and Barry, John R.},
  journal={IEEE Transactions on Wireless Communications}, 
  title={Fast maximum-likelihood decoding of the golden code}, 
  year={2010},
  volume={9},
  number={1},
  pages={26-31},
  doi={10.1109/TWC.2010.01.081512}}

@ARTICLE{srinath2011silver,
  author={Srinath, K. Pavan and Rajan, B. Sundar},
  journal={IEEE Transactions on Information Theory}, 
  title={Generalized Silver Codes}, 
  year={2011},
  volume={57},
  number={9},
  pages={6134-6147},
  doi={10.1109/TIT.2011.2162276}}

@techreport{itur_m2160,
  author      = {{ITU-R}},
  title       = {{Recommendation ITU\-R M\.2160\-0}: Framework and overall objectives of the future development of {IMT} for 2030 and beyond},
  institution = {International Telecommunication Union -- Radiocommunication Sector},
  number      = {M.2160-0},
  year        = {2023},
  month       = nov,
  url         = {https://www.itu.int/dms_pubrec/itu-r/rec/m/R-REC-M.2160-0-202311-I-21-21PDF-E.pdf},
  note         = {Accessed: 2025-12-18}
}

@techreport{itur_m2516,
  author      = {{ITU-R}},
  title       = {{Report ITU-R M.2516-0}: Future technology trends of terrestrial International Mobile Telecommunications systems towards 2030 and beyond},
  institution = {International Telecommunication Union -- Radiocommunication Sector},
  number      = {M.2516-0},
  year        = {2022},
  month       = nov,
  url         = {https://www.itu.int/dms_pub/itu-r/opb/rep/R-REP-M.2516-2022-PDF-E.pdf},
  note         = {Accessed: 2025-12-18}
}

@article{freudenberger2017eisenstein,
  author  = {Freudenberger, J{\"u}rgen and Shavgulidze, Sergo},
  title   = {Signal Constellations Based on {E}isenstein Integers for Generalized Spatial Modulation},
  journal = {IEEE Communications Letters},
  volume  = {21},
  number  = {3},
  pages   = {556--559},
  year    = {2017},
  month   = mar,
  doi     = {10.1109/LCOMM.2016.2630701}
}

@INPROCEEDINGS{sinnokrot2008alamouti,
  author={Sinnokrot, Mohanned O. and Barry, John R. and Madisetti, Vijay K.},
  booktitle={2008 42nd Asilomar Conference on Signals, Systems and Computers}, 
  title={Embedded Alamouti space-time codes for high rate and low decoding complexity}, 
  year={2008},
  volume={},
  number={},
  pages={1749-1753},
  doi={10.1109/ACSSC.2008.5074726}}

@ARTICLE{ramos2014fading,
  author={Ramos, Alberto Gil C. P. and Rodrigues, Miguel R. D.},
  title={Fading Channels With Arbitrary Inputs: Asymptotics of the Constrained Capacity and Information and Estimation Measures},
  journal={IEEE Transactions on Information Theory},
  year={2014},
  volume={60},
  number={9},
  pages={5653--5672},
  month=sep,
  doi={10.1109/TIT.2014.2333746}
}

@techreport{Samsung2026Modulation6GR,
  author      = {{Samsung}},
  title       = {Discussion on Modulation for {6GR}},
  institution = {{3GPP}},
  type        = {{TSG RAN WG1 Meeting \#124} contribution},
  number      = {R1-2600753},
  address     = {Gothenburg, Sweden},
  year        = {2026},
  month       = feb,
  note        = {February 9--13, 2026},
  url         = {https://www.3gpp.org/ftp/tsg_ran/wg1_rl1/TSGR1_124/Docs/R1-2600753.zip}
}

@inproceedings{Shao2026APSKISAC,
  author    = {Shao, Yujie and Qiu, Min and Lee, Ming-Chun and
               Huang, Yu-Chih and Yuan, Jinhong},
  title     = {Design of {APSK} Constellations Approaching the
               Communication-Sensing Pareto Boundary for {ISAC}},
  booktitle = {2026 IEEE Information Theory Workshop (ITW)},
  year      = {2026},
  note      = {Accepted; extended version available as arXiv:2605.25047},
  url       = {https://arxiv.org/abs/2605.25047}
}

@article{Du2024ISACPCS,
  author  = {Du, Zhen and Liu, Fan and Xiong, Yifeng and
             Han, Tony Xiao and Eldar, Yonina C. and Jin, Shi},
  title   = {Reshaping the {ISAC} Tradeoff Under {OFDM} Signaling:
             A Probabilistic Constellation Shaping Approach},
  journal = {IEEE Transactions on Signal Processing},
  volume  = {72},
  pages   = {4782--4797},
  year    = {2024},
  doi     = {10.1109/TSP.2024.3465499}
}

@inproceedings{Geiger2025Joint,
  author    = {Geiger, Benedikt and Liu, Fan and Lu, Shihang and
               Rode, Andrej and Schmalen, Laurent},
  title     = {Joint Optimization of Geometric and Probabilistic
               Constellation Shaping for {OFDM-ISAC} Systems},
  booktitle = {2025 IEEE 5th International Symposium on Joint
               Communications \& Sensing (JC\&S)},
  year      = {2025},
  month     = jan,
  pages     = {1--6},
  address   = {Oulu, Finland},
  doi       = {10.1109/JCS64661.2025.10880638}
}

\end{document}